%% file: main.tex
\documentclass[conference,compsoc]{IEEEtran}
\IEEEoverridecommandlockouts

\usepackage{xcolor}
\usepackage{graphicx}
\ifCLASSOPTIONcompsoc
  \usepackage[nocompress]{cite}
\else
  \usepackage{cite}
\fi
\usepackage[utf8]{inputenc}
\usepackage{amsmath}
\usepackage{amsfonts}
\usepackage{mathtools}
\usepackage{tabularx}
\usepackage{booktabs}
\usepackage{xspace}
\usepackage{boxedminipage}
\usepackage[inline]{enumitem}
\usepackage{amsthm}
\usepackage{multirow}
\usepackage{multicol}
\usepackage{pifont}
\usepackage{makecell}
\usepackage{adjustbox}
\usepackage{algorithm}
\usepackage{thmtools}
\usepackage{stfloats}
\usepackage[noend]{algpseudocode}
\ifCLASSOPTIONcompsoc
  \usepackage[caption=false,font=footnotesize,labelfont=sf,textfont=sf]{subfig}
\else
  \usepackage[caption=false,font=footnotesize]{subfig}
\fi
\usepackage{float}
\usepackage{tikz}
\usetikzlibrary{arrows.meta,positioning,calc,fit,backgrounds}
\usepackage{soul}
\usepackage{listings}
\usepackage{threeparttable}
\usepackage{balance}
\usepackage{xurl}
\usepackage[hidelinks]{hyperref}
\usepackage[capitalize,noabbrev]{cleveref}

\input{setup}

\algdef{SE}[EVENT]{Upon}{EndUpon}[1]{\textbf{upon} #1 \textbf{do}}{\textbf{end upon}}
\algtext*{EndUpon}

\newcommand{\paperauthorblock}[3]{%
\begin{minipage}[t]{0.29\textwidth}\centering
{\normalfont\normalsize #1}\\[0.25ex]
{\normalfont\itshape\normalsize #2\\#3}
\end{minipage}}

\title{\name: Practical Low-Latency Signature-Free BFT Consensus}
\author{
\begin{tabular}{@{}c@{\hspace{0.025\textwidth}}c@{\hspace{0.025\textwidth}}c@{}}
\multicolumn{3}{c}{
\paperauthorblock{Qianyu Yu}
{The Hong Kong University of\\Science and Technology\\(Guangzhou)}
{qyu100@connect.hkust-gz.edu.cn}
\hspace{0.08\textwidth}
\paperauthorblock{Juan Villacis}
{University of Bern}
{juan.villacis@unibe.ch}}
\\[1.8em]
\paperauthorblock{Giuliano Losa}
{Stellar Development Foundation}
{giuliano@stellar.org}
&
\paperauthorblock{Zhuolun Xiang}
{Aptos Labs}
{xiangzhuolun@gmail.com}
&
\paperauthorblock{Xuechao Wang\textsuperscript{*}\thanks{\textsuperscript{*}Correspondence: Xuechao Wang.}}
{The Hong Kong University of\\Science and Technology\\(Guangzhou)}
{xuechaowang@hkust-gz.edu.cn}
\end{tabular}
}

\begin{document}

\maketitle

\begin{abstract}
\input{abstract}
\end{abstract}

\IEEEpeerreviewmaketitle

\input{introduction}

\input{preliminary}
\input{protocol}

\input{security_analysis2}
\input{evaluation}
\input{related_work}
\input{conclusion}
\ifCLASSOPTIONcompsoc
  \section*{Acknowledgments}
\else
  \section*{Acknowledgment}
\fi
Xuechao Wang is supported by the Guangdong Basic and Applied Basic Research Foundation (No. 2025A1515110040), the Guangdong Provincial Key Lab of Integrated Communication, Sensing and Computation for Ubiquitous Internet of Things (No. 2023B1212010007), the Guangzhou-HKUST(GZ) Joint Funding Program (No. 2025A03J3882), the Guangzhou Municipal Science and Technology Project (No. 2025A04J4168), and a gift from Stellar Development Foundation.

We thank Nibesh Shrestha and Balaji Arun for helpful advice on the implementation.

\IEEEtriggeratref{40}
\IEEEtriggercmd{\balance}
\bibliographystyle{IEEEtranS}
\bibliography{main}

\end{document}

%% file: setup.tex
\lstdefinestyle{pseudostyle}{
  basicstyle=\ttfamily\footnotesize,
  columns=fullflexible,
  keepspaces=true,
  breaklines=true,
  showstringspaces=false,
  tabsize=2
}

\newcounter{algorithmicH}
\AtBeginEnvironment{algorithmic}{\stepcounter{algorithmicH}}
\makeatletter
\providecommand*{\theHALG@line}{}
\renewcommand{\theHALG@line}{\arabic{algorithmicH}.\arabic{ALG@line}}
\makeatother

\newcommand{\parties}{\mathcal{P}}
\newcommand{\party}[1]{\ensuremath{p_{#1}}\xspace}

\newcommand{\tuple}[1]{\langle #1 \rangle}

\newcommand{\True}{\mathsf{true}}
\newcommand{\False}{\mathsf{false}}

\algrenewcommand\textproc{}

\definecolor{yescolor}{HTML}{026378}
\def\yes{\textcolor{yescolor}{\checkmark}}
\def\no{\textcolor{red}{\ding{55}}}

\newcommand{\name}{Simple-IT\xspace}
\newcommand{\shortname}{\ensuremath{\mathcal{S}}\xspace}

\makeatletter
\renewcommand{\ALG@name}{}
\makeatother

\algrenewcommand\textproc{}
\algdef{SE}[EVENT]{Event}{EndEvent}[1]{\textbf{upon}\ #1\ \algorithmicdo}{\algorithmicend\ \textbf{event}}%
\algtext*{EndEvent}

\newcommand{\Dto}{\Delta_{\text{to}}}

\newcommand{\Commit}{\mathsf{commit}}
\newcommand{\Accept}{\mathsf{accept}}

\newcommand{\var}[1]{\mathsf{#1}}
\newcommand{\fn}[1]{\mathbf{\mathsf{#1}}}

\newcommand{\iunderline}[1]{\noindent\underline{#1}}

\newcommand{\floor}[1]{\lfloor #1 \rfloor}
\newcommand{\ceil}[1]{\lceil #1 \rceil}

\newtheorem{lemma}{Lemma}

\newtheorem{corollary}{Corollary}
\newtheorem{theorem}{Theorem}
\newtheorem{definition}{Definition}

\newcommand{\sig}[1]{\langle #1 \rangle}

\newcommand{\Vote}{\mathsf{vote}}

\renewcommand{\paragraph}[1]{\medskip\noindent\textbf{#1}}

\newif\iffull

%% file: abstract.tex
Recent advances in quantum computing pose a looming threat to most current Byzantine fault-tolerant (BFT) consensus protocols, which rely on post-quantum-vulnerable public-key signature schemes such as Ed25519 and BLS.
Instead of switching to much more expensive post-quantum-secure signature schemes, an alternative is to use signature-free protocols, which rely only on cheap, post-quantum-secure authenticated channels.

In this paper, we ask whether signature-free BFT consensus protocols can match the performance of current state-of-the-art, post-quantum-vulnerable BFT consensus protocols.
While previous work on the Sailfish++ protocol showed that state-of-the-art throughput is attainable signature-free, the question of latency is still open.
Several recent signature-free, leader-based protocols have low latency in theory, but, to our knowledge, none has been implemented and evaluated in practice.
In this work, we propose Simple-IT, a new leader-based, signature-free BFT consensus protocol that achieves a theoretical latency of 4 message delays (one more than the optimum), and only 3 on its optimistic path.
Crucially, Simple-IT is simple enough to be amenable to implementation and to practical optimizations such as speculative pipelining.
As we show experimentally in a geo-distributed testbed, it achieves both throughput and latency competitive with state-of-the-art post-quantum-vulnerable protocols.

%% file: introduction.tex
\section{Introduction}

Partially synchronous~\cite{dwork1988consensus} Byzantine fault-tolerant consensus protocols (BFT protocols, for short) are at the core of many high-performance reliable replicated systems such as blockchains, and their security and performance are critical.
However, progress in quantum computing poses a credible, near-term threat to their security.

Virtually every deployed BFT protocol authenticates parties' messages with public-key signatures such as Ed25519 signatures or other schemes based on discrete-logarithm and factoring problems.
However, a quantum computer running Shor's algorithm can break such schemes and recover a party's secret key from its public key in polynomial time, allowing the attacker to forge messages and break both the safety and liveness of the protocol.
Building a quantum computer capable of such an attack remains a major engineering challenge, but the estimated costs keep falling~\cite{babbush2026ecc}, and NIST and the EU have called for transitioning to post-quantum-secure systems by 2035 and 2030, respectively~\cite{moody2024transition,nis2025pqcRoadmap}.

One solution to secure BFT protocols is to adopt quantum-secure signatures.
Unfortunately, signing and verifying signatures are on the critical path of BFT protocols, and quantum-secure signatures are much more expensive than their post-quantum-vulnerable counterparts.
For example, signatures from NIST-standardized schemes~\cite{fips204,fips205} are at least one order of magnitude larger than the post-quantum-vulnerable Ed25519 signatures, and some are also much slower to verify.
Our experiments in~\Cref{sec: evaluation} show that, in a state-of-the-art BFT protocol, switching to post-quantum-secure ML-DSA-65 signatures~\cite{fips204,ducas2018crystals} doubles latency.

Another solution is to design signature-free BFT protocols, which, after a setup phase, rely only on authenticated channels secured using fast and quantum-secure symmetric cryptography such as HMAC-SHA256 (HMAC-SHA256 is only vulnerable to Grover's algorithm, whose quadratic speedup is offset by a modest increase in key size).

However, with only authenticated channels, we lose the transferability property of signatures (i.e., a party \(p_1\) cannot prove to another party \(p_2\) that a third party \(p_3\) sent a given message), and we must instead rely on mechanisms like Bracha's reliable broadcast~\cite{bracha1987asynchronous} to ensure parties converge on the same view of the system's execution.
Naively applying such mechanisms adds latency and increases communication complexity, while trying to orchestrate them cleverly can lead to overly intricate algorithms that are hard to turn into practical BFT protocols.

In this paper, we ask whether partially synchronous, signature-free BFT consensus protocols are nevertheless a practical solution to obtaining post-quantum security with optimal resilience (\(n>3f\)) while matching the performance of the best post-quantum-vulnerable BFT protocols.

Recently, Shrestha et al.~\cite{shresthaOptimisticSignatureFreeReliable2025a} showed that, if one is willing to compromise on latency and scalability for better throughput, the answer is yes.
They propose Sailfish++, a DAG-based signature-free BFT protocol, and show on a geo-distributed testbed that Sailfish++ matches the throughput of Sailfish~\cite{shrestha2024sailfish}, which itself achieves state-of-the-art throughput among post-quantum-vulnerable BFT protocols. 
However, neither Sailfish++ nor Sailfish matches the latency or scalability of leader-based post-quantum-vulnerable BFT protocols.
Sailfish++ has an average optimistic commit latency of roughly 5.6 message delays\footnote{\label{fn:sailfish-latency}Leader vertices commit in 3 message delays, \(n-f-1\) other vertices commit in 5, and the remaining \(f\) commit in 7 message delays.} and sends \(O(n^3)\) bits per view, whereas some leader-based BFT protocols, e.g. DispersedSimplex~\cite{shoupSingSongSimplex2024}, achieve a good-case commit latency of only 3 message delays and send only \(O(n^2)\) bits per view.

\begin{table*}[t]
    \caption{Leader-based BFT protocols (without block dissemination). Simple-IT entries use non-speculative proposals.}%
    \label{table:leader-based-comparison}
    \setlength\tabcolsep{4pt}
    \begin{center}
    \newcommand{\hdii}[2]{\textbf{\makecell{#1\\#2}}}
    \newcommand{\hdiii}[3]{\textbf{\makecell{#1\\#2\\#3}}}
    \begin{threeparttable}
    \begin{tabular}{l c c c c}
    \toprule
      & \hdii{Good-case/optimistic}{commit latency}
      & \hdii{Eventual worst-case}{view duration}
      & \hdii{Bits sent}{per view}
      & \textbf{Signature-free}
      \\ [0.5ex]
    \midrule
      IT-HS~\cite{abrahamInformationTheoreticHotStuff2021}         & \(6\delta\) & \(11\Delta+\delta\) & $O(n^2)$ & \yes \\
      Alg.\ BFT~\cite[Ch.~3]{castro_thesis}\tnote{\dag}               & \(3\delta\) & \(7\Delta+5\delta\)\tnote{\ddag} & $O(n^3)$ & \yes \\
      TetraBFT~\cite{yuTetraBFTReducingLatency2024b}               & \(5\delta\) & \(8\Delta+\delta\)\tnote{\S} & $O(n^2)$ & \yes \\
      Forget-IT~\cite{abraham2026forget} & \(3\delta\) & \(4\Delta+3\delta\) & $O(n^2)$ & \yes \\
      Simplex~\cite{chanSimplexConsensusSimple2023} & \(3\delta\) & \(3\Delta+\delta\) & $O(n^2)$ & \no \\
    \midrule
      \textbf{\name} (Opt-RBC)  & \(5\delta\)/\(3\delta\)\tnote{*} & \(8\Delta+4\delta\) & $O(n^2)$ & \yes \\
      \textbf{\name} (Bracha-RBC) & \(4\delta\) & \(5\Delta+2\delta\) & $O(n^2)$ & \yes \\
    \bottomrule
    \end{tabular}
    \begin{tablenotes}
    \item[*] The optimistic bound assumes \(GST=0\), correct previous and current leaders, and at most \(f_o=\floor{n/2}-f\) Byzantine parties (\Cref{def: optimistic-commit-latency}).
    \item[\dag] No peer-reviewed versions are available.
    \item[\ddag] Our estimate, as the thesis does not specify timeouts in terms of \(\Delta\).
    \item[\S] With the view timer set to \(8\Delta\), the minimum needed for liveness; the paper prescribes \(9\Delta\).
    \end{tablenotes}
    \end{threeparttable}
    \end{center}
\end{table*}

\begin{table*}[t]
    \caption{Consensus protocols with practical, high-performance implementations including block dissemination}%
    \label{table:practical-comparison}
    \setlength\tabcolsep{7.5pt}
    \begin{center}
    \newcommand{\hd}[2]{\textbf{\makecell{#1\\#2}}}
    \newcommand{\hdiii}[3]{\textbf{\makecell{#1\\#2\\#3}}}
    \begin{threeparttable}
    \begin{tabular}{l c c c c c}
    \toprule
      & \hdiii{Good-case}{/optimistic}{commit latency}\tnote{\P}
      & \hdiii{Eventual}{worst-case}{view duration}
      & \hd{View complexity}{(bits)}
      & \hdiii{Optimistic}{amortized}{block time}\tnote{\#}
      & \hd{Signature-}{free} \\ [0.5ex]
    \midrule
      Autobahn~\cite{giridharanAutobahnSeamlessHigh2024}                           & \(8\delta/6\delta\)\tnote{*} & \(10\Delta+\delta\) & \(O(n^3)\)& \(3\delta\) & \no \\
      DispersedSimplex~\cite{shoupSingSongSimplex2024}                           & \(3\delta\) & \(3\Delta+\delta\) & \(O(n^2)\)& \(2\delta\)& \no \\
      Sailfish~\cite{shrestha2024sailfish}                           & \(5.6\delta\)\tnote{\dag} & \(4\Delta+2\delta\)\tnote{\ddag} & \(O(n^3)\) & \(2\delta\) & \no \\
      Sailfish++~\cite{shrestha2025optimistic} (Bracha-RBC)                           & \(7.9\delta\)\tnote{\dag} & \(5\Delta+4\delta\)\tnote{\S} & $O(n^3)$ & \(3\delta\) & \yes  \\
      Sailfish++~\cite{shrestha2025optimistic} (Opt-RBC)                           & \(10.2\delta/5.6\delta\)\tnote{\dag} & \(8\Delta+5\delta\)\tnote{\ddag} & $O(n^3)$ & \(2\delta\) & \yes \\
    \midrule
      Dispersed-\name (Opt-RBC), non-speculative & \(5\delta/3\delta\) & \(8\Delta+4\delta\) & $O(n^2)$ & \(2\delta\) & \yes \\
      Dispersed-\name (Opt-RBC), speculative~\ding{72} & ---/\(3\delta\)\tnote{\textbardbl} & \(8\Delta+4\delta\) & $O(n^2)$ & \(1\delta\) & \yes \\
      Dispersed-\name (Bracha-RBC), non-speculative & \(4\delta\) & \(5\Delta+2\delta\) & $O(n^2)$ & \(3\delta\) & \yes \\
      Dispersed-\name (Bracha-RBC), speculative~\ding{72} & ---/\(4\delta\)\tnote{\textbardbl} & \(5\Delta+2\delta\) & $O(n^2)$ & \(1\delta\) & \yes \\
      Mempool-\name (Opt-RBC), non-speculative~\ding{72} & \(7\delta/5\delta\) & \(8\Delta+4\delta\) & $O(n^2)$ & \(2\delta\) & \yes \\
      Mempool-\name (Bracha-RBC), non-speculative~\ding{72} & \(6\delta\) & \(5\Delta+2\delta\) & $O(n^2)$ & \(3\delta\) & \yes \\
    \bottomrule
    \end{tabular}
    \begin{tablenotes}
    \item[\P] For Simple-IT, optimistic commit bounds use~\Cref{def: optimistic-commit-latency}, with \(f_o=f\) for Bracha-RBC and \(f_o=\floor{n/2}-f\) for Opt-RBC.
    \item[\#] Block-time entries are amortized upper bounds. For Simple-IT, see~\Cref{def: optimistic-amortized-block-time,lemma: optimistic-block-time-delivery,lemma: optimistic-block-time-pipelined}.
    \item[*] \(6\delta/4\delta\) with optimistic tips.
    \item[\dag] Average over the vertices of a DAG round: the leader vertex commits in \(1\text{RBC}+1\delta\), \(n-f-1\) vertices in \(2\text{RBC}+1\delta\), and the remaining \(f\) vertices in \(3\text{RBC}+1\delta\).
    \item[\ddag] Our estimate, as the paper does not prescribe timeouts for this RBC instantiation.
    \item[\S] Using the paper's prescribed \(5\Delta\) timeout, which applies to views that follow a correct leader; views entered via timeout certificate use \(8\Delta\), giving \(8\Delta+4\delta\).
    \item[\textbardbl] No bounded good-case commit latency (\Cref{def: good-case-commit-latency}): a speculative proposal that follows a Byzantine leader can be aborted (\Cref{sec: variants}). The optimistic bound is \Cref{lemma: optimistic-latency}.
    \item[\ding{72}] \name variants evaluated in~\Cref{sec: evaluation}. The speculative variants use pipeline depth \(k=2\) with Opt-RBC and \(k=3\) with Bracha-RBC.
    \end{tablenotes}
    \end{threeparttable}

    \end{center}
\end{table*}

We therefore focus on signature-free leader-based BFT protocols and ask whether they can reach the same throughput, latency, and scalability as their signature-based counterparts.
IT-HS~\cite{abrahamInformationTheoreticHotStuff2021}, TetraBFT~\cite{yuTetraBFTReducingLatency2024b}, and Forget-IT~\cite{abraham2026forget} achieve a good-case latency of 6, 5, and 3 message delays, respectively.
However, these results remain theoretical: to our knowledge, none of these protocols have been implemented and evaluated in practice.
Whether signature-free leader-based protocols can combine low commit latency, efficient block
dissemination, and practical optimizations to match the performance of their signature-based
counterparts therefore remains an open question.

To answer this question, we present the \name protocol, a new partially synchronous, signature-free, leader-based BFT consensus protocol with optimal resilience of \(n>3f\).
At a high level, \name adapts the Simplex~\cite{chanSimplexConsensusSimple2023} protocol to the signature-free setting by replacing quorum certificates with convergence mechanisms that make them unnecessary.
In Simplex's parlance, \name uses reliable broadcast (RBC) to notarize blocks and, in the good case,  tob-delivers them after one further message delay, giving a commit latency of \(1\text{RBC}+1\delta\), where \(\delta\) upper bounds message delay.
It uses reliable notification (RN), a novel abstraction introduced in this paper, to disable consensus rounds that have timed out.
Both primitives guarantee that all correct parties converge on the same outcome, and thus parties do not need certificates to prove outcomes to each other.

The first strength of \name is its latency.
Paired with Bracha's RBC protocol~\cite{bracha1987asynchronous}, which we call Bracha-RBC, \name achieves a good-case latency of 4 message delays;
paired with the optimistic reliable-broadcast protocol of Shrestha et al.~\cite{shresthaOptimisticSignatureFreeReliable2025a,shrestha2025optimistic}\footnote{The original version~\cite{shresthaOptimisticSignatureFreeReliable2025a} contains an error in the RBC protocol which is fixed in the arXiv version~\cite[version v4]{shrestha2025optimistic}.}, which we call Opt-RBC, \name achieves a latency of only 3 message delays if at most approximately \(\frac{n}{6}\) parties\footnote{The exact number is \(\floor{\frac{n}{2}}-f\); see~\Cref{sec: rbc-impl-latency}.} are Byzantine (when \(n\approx 3f\)) and previous leaders are correct (we call this the optimistic case).
Note that, under optimal resilience, three message delays is the lower bound for tolerating more than two failures~\cite{kuznetsovRevisitingOptimalResilience2021}\footnote{Kuznetsov et al.~\cite{kuznetsovRevisitingOptimalResilience2021} show that any protocol with a 2-message-delay fast path tolerating \(t\) crash failures must have Byzantine resilience of \(n\geq \max(3f+1,3f+2t-1)\); this implies that, in an optimally-resilient protocol (\(f=\floor{\frac{n-1}{3}}\)), a 2-message-delay fast path cannot tolerate more than 2 failures.}, even for signature-based protocols, and \name achieves it at the cost of slightly bigger optimistic-path quorums (roughly \(\frac{5n}{6}\)) than signature-based protocols (roughly \(\frac{2n}{3}\)).
A latency comparison with other leader-based BFT protocols appears in~\Cref{table:leader-based-comparison}.

Second, like Simplex, \name follows simple rules that make it easily amenable to practical optimizations such as pipelining and speculative proposals.
These techniques allow \name to maximize network utilization by running different consensus stages in parallel, and they reduce its optimistic amortized block time, i.e.\ the average time between delivered blocks, to at most one message delay.
The latter is crucial to reducing the end-to-end latency of transactions arriving between two blocks. 
To demonstrate the flexibility afforded by \name's simple rules, we present four variants of \name, each corresponding to a choice of RBC algorithm (Bracha-RBC or Opt-RBC) and to whether to use a two-stage RBC/commit pipeline or to also speculatively pipeline proposals.

To determine whether \name can reach performance competitive with post-quantum-vulnerable protocols in practice, we develop two implementations of \name: Mempool-\name and Dispersed-\name.
Mempool-\name implements the base protocol, without speculative proposals, and uses a signature-free variant of the Autobahn~\cite{giridharanAutobahnSeamlessHigh2024} mempool to disseminate blocks.
Dispersed-\name implements the speculative variant of \name, but does not use a mempool; instead, leaders disseminate their blocks with an erasure-coded RBC protocol based on well-known techniques proposed by Cachin and Tessaro for AVID protocols~\cite{cachin2005asynchronous}.
Both implementations can be configured with Bracha-RBC or Opt-RBC (both erasure coded in the case of Dispersed-\name).
Together, these implementations cover two practical approaches to data dissemination in high-throughput, leader-based BFT consensus deployed in production: shared-mempool-style dissemination, as used in systems such as Aptos Quorum Store~\cite{quorum-store}, and leader-driven erasure-coded block propagation, as used in systems such as Monad RaptorCast~\cite{monad} and Solana Turbine~\cite{turbine}.
\Cref{table:practical-comparison} compares the theoretical characteristics of Mempool-\name and Dispersed-\name with those of other protocols that have practical, high-performance implementations.

We evaluate Mempool-\name and Dispersed-\name on a real-world testbed of 50 nodes spread across 5 regions over 3 continents, and we compare their throughput and latency against their post-quantum-vulnerable counterparts Autobahn and DispersedSimplex~\cite{shoupSingSongSimplex2024}.
Both Autobahn and DispersedSimplex achieve state-of-the-art performance in their respective categories.

The experimental results show that both Mempool-\name and Dispersed-\name closely match or surpass their post-quantum-vulnerable counterparts.
Mempool-\name achieves better latency and throughput than Autobahn across the entire latency-throughput spectrum, reaching 170,000 transactions per second (approximately 87 MB/s) at roughly 0.7 seconds of latency.
Dispersed-\name matches DispersedSimplex's throughput before saturation (roughly 120,000 transactions per second, or 61 MB/s), and its latency in that regime is at most 10 percent higher.

The results show that we can answer our motivating question affirmatively: signature-free, leader-based BFT protocols can be simple enough to implement efficiently and closely match, in practice, the throughput and latency of the best post-quantum-vulnerable leader-based BFT protocols.
Thus, post-quantum security for high-performance BFT protocols does not require expensive post-quantum-secure signatures on the critical path.

\paragraph{Roadmap.}
We first define the model, total-order broadcast, and the reliable broadcast and reliable notification abstractions that \name uses in~\Cref{sec: preliminaries}.
We then present the \name protocol, including its optimistic and speculative variants, in~\Cref{sec: protocol}, and analyze its safety, liveness, latency, and block time in~\Cref{sec: security}.
\Cref{sec: evaluation} describes the Mempool-\name and Dispersed-\name implementations and evaluates them on a geo-distributed testbed.
Finally, \Cref{sec: related work} discusses related work and \Cref{sec: conclusion} concludes.

%% file: preliminary.tex
\section{Preliminaries}%
\label{sec: preliminaries}

\subsection{Distributed-Computing Model}

We consider a set of \(n\) parties $\parties = \{\party{1},\ldots,\party{n}\}$ subject to $f<\frac{n}{3}$ Byzantine failures and executing a protocol in an eventually-synchronous message-passing system with reliable, authenticated channels.

Each party is either correct or Byzantine, and there are at most $f<\frac{n}{3}$ Byzantine parties.%
Correct parties follow the protocol; Byzantine parties may behave arbitrarily.
Which parties are Byzantine is unknown to the protocol.
As is customary, we say that a set of \(n-f\) parties is a quorum, noting the quorum intersection property: every two quorums have a correct party in common.

Parties communicate by message passing and we assume that there is an authenticated, reliable channel between every pair of parties.
This means that every message sent by a correct party to a correct party is eventually received,
and that if a party \party{i} receives a message \(m\) on the authenticated channel established with a party \party{j}, then \party{i} can be sure that \party{j} sent \(m\).
We denote by \(\delta\) the maximum message delay over a period under consideration depending on context.

We also assume that the system is eventually synchronous~\cite{dwork1988consensus}.
This means that there is an unknown \emph{Global Stabilization Time} (GST) and a known constant \(\Delta\) such that every message sent by a correct party to a correct party at time \(t\) is received by time \(\max(\text{GST},t)+\delta\), where \(\delta< \Delta\) is the (unknown) maximum message delay after GST.
Additionally, we assume that parties have local clocks which, after GST, progress at the same speed.

\subsection{Problem Definitions}
The \name protocol variants formally implement Total-Order Broadcast, as defined below.
\begin{definition}[Total-Order Broadcast]%
    \label{def: tob}
    In a total-order broadcast (TOB) protocol, each party may tob-submit blocks (of transactions) by calling tob\_submit$(b)$, where \(b\) is a block, and the TOB protocol may tob-deliver blocks by calling tob\_deliver$(b)$.
    A total-order broadcast protocol must satisfy the following properties:
    \begin{itemize}[noitemsep,leftmargin=*]
        \item[-] \textbf{Total Order.} If a correct party tob-delivers block \(b\) before block \(b'\), then no correct party tob-delivers \(b'\) before \(b\).
        \item[-] \textbf{Totality.} If a correct party tob-delivers a block \(b\), then every correct party eventually tob-delivers \(b\).
        \item[-] \textbf{Liveness.} If every correct party tob-submits infinitely many blocks, then every correct party tob-delivers infinitely many blocks tob-submitted by correct parties.
    \end{itemize}
\end{definition}

Note that Liveness is intentionally weaker than some textbook liveness properties (e.g.\ the validity property in \cite[Chapter 6.1]{cachinIntroductionReliableSecure2011}); it rules out trivial protocols and allows abstracting over implementation details like whether to queue, retry, or forward blocks to other parties.
The protocols we present satisfy stronger progress properties that depend on the protocol.
For example, \name ensures that, after GST, a non-speculative proposal by a correct leader is tob-delivered by all correct parties.
With speculative proposals and a speculative pipeline of depth \(k=3\), \name ensures that, after GST and after 4 consecutive rounds with correct leaders, a correct proposal in the fifth round is tob-delivered by all correct parties.

Our protocols are leader-based, and we define latency metrics as follows, expressed in terms of \(\delta\) (e.g.\ the good-case commit latency of \name with non-speculative proposals and Bracha-RBC is \(4\delta\)).
\begin{definition}[Good-case Commit Latency]
    \label{def: good-case-commit-latency}
    The good-case commit latency is the maximum time it takes in a round started after GST for the proposal of a correct leader to be tob-delivered by all correct parties.
\end{definition}

\begin{definition}[Optimistic Commit Latency and Optimistic Resilience \(f_o\)]
    \label{def: optimistic-commit-latency}
    The optimistic commit latency of a leader-based protocol is the maximum time it takes for the proposal of a correct leader to be tob-delivered by all correct parties, assuming that all previous leaders and the current leader are correct, that $GST=0$, and that there are at most $f_o$ Byzantine parties.
    The constant \(f_o\) is a protocol-specific optimistic resilience parameter.
\end{definition}
In this paper, we use \(f_o=f\) for protocols based on Bracha-RBC and \(f_o = \floor{n/2}-f\) for protocols based on the Opt-RBC protocol of Shrestha et al.~\cite{shrestha2025optimistic}.
For example, for protocols based on Opt-RBC, when \(f\approx n/3\) we have \(f_o\approx n/6\).

\begin{definition}[Optimistic Amortized Block Time]
    \label{def: optimistic-amortized-block-time}
    The optimistic amortized block time is the worst-case long-run average time per block tob-delivered by a correct party, under the optimistic conditions of~\Cref{def: optimistic-commit-latency}.
\end{definition}

\begin{definition}[Eventual Worst-Case View Duration]%
    \label{def: view-duration}
    The eventual worst-case view duration (or view duration for short) of a protocol is the maximum time, in a view started after GST, between the last correct party entering a view and the last correct party leaving it.
\end{definition}

We also make use of two abstractions, reliable broadcast and reliable notification, in order to give modular protocol descriptions.
The reliable notification problem is novel and may be of more general interest.

\begin{definition}[Reliable Broadcast~\cite{bracha1987asynchronous}]
    \label{def: reliable-broadcast}
    In a reliable broadcast (RBC) protocol, a party \(p\) may rb-broadcast a message $m=\tuple{p,r,d}$, for some identifier \(r\) (in \name, a round number) and some payload \(d\), by calling rb\_broadcast$(m)$; the RBC protocol may rb-deliver messages of the form $\tuple{p,r,d}$, where \(p\) is a party, \(r\) an identifier, and \(d\) the payload, by calling rb\_deliver$(m)$.
    A reliable broadcast protocol must satisfy the following properties:
    \begin{itemize}[noitemsep,leftmargin=*]
        \item[-] \textbf{Uniqueness.} If a correct party rb-delivers a message \(m=\tuple{p,r,d}\), then it does not rb-deliver any other message \(m'=\tuple{p,r,d'}\) with the same sender and same identifier and does not rb-deliver \(m\) again.
        \item[-] \textbf{Validity.} If a correct party \(p\) rb-broadcasts a message $m=\tuple{p,r,d}$, for some \(r\) and \(d\), then every correct party eventually rb-delivers \(m\).
        \item[-] \textbf{Totality.} If a correct party rb-delivers a message \(m\), then every correct party eventually rb-delivers \(m\).
    \end{itemize}
\end{definition}

\begin{definition}[Reliable Notification]%
    \label{def: reliable-notification}
    In a reliable notification (RN) protocol, each party may rn-raise a flag \(e\) by calling rn\_raise\((e)\) and the protocol may rn-confirm \(e\) by calling rn\_confirm\((e)\).
    A reliable notification protocol must satisfy the following properties:
\begin{itemize}[noitemsep,leftmargin=*]
    \item[-] \textbf{Unanimity.} If every correct party rn-raises \(e\), then every correct party eventually rn-confirms \(e\).
    \item[-] \textbf{Totality.} If a correct party rn-confirms \(e\), then every correct party eventually rn-confirms \(e\).
    \item[-] \textbf{Validity.} If a correct party rn-confirms \(e\), then at least $n-2f$ correct parties rn-raised \(e\).
\end{itemize}
\end{definition}

\subsection{Latency of RBC Implementations}%
\label{sec: rbc-impl-latency}

For the purpose of analyzing latency, several parameters of RBC protocols are important.

\begin{definition}[Optimistic RBC delay \(d_o\)]
    \label{def: opt-rbc-delay}
    The optimistic delay of a reliable broadcast protocol, denoted by $d_o$, is the maximum number of message delays between the moment at which a correct sender rb-broadcasts a message and the moment at which it is rb-delivered by all correct parties, over all executions in which at most \(f_o\) parties are Byzantine.
\end{definition}

\begin{definition}[Worst-case RBC delay \(d_s\)]%
    \label{def: bad-rbc-delay}
    The worst-case delay of a reliable broadcast protocol, denoted by $d_s$, is the maximum number of message delays between the moment at which a correct sender rb-broadcasts a message and the moment at which it is rb-delivered by all correct parties.
\end{definition}

\begin{definition}[RBC totality delay \(d_t\)]
    The totality delay of a reliable broadcast protocol, denoted by $d_t$, is the maximum number of message delays between the first moment at which a correct party rb-delivers a message and the moment at which it is rb-delivered by all other correct parties.
\end{definition}

In this work, we make use of two reliable broadcast implementations: Bracha's broadcast protocol~\cite{bracha1987asynchronous}, which we call Bracha-RBC, and the optimistic RBC protocol of Shrestha et al.~\cite{shrestha2025optimistic}, which we call Opt-RBC.
Both implement reliable broadcast assuming \(n>3f\), and their latency characteristics appear in~\Cref{tab:broadcast_comparison}.
Under our channel convention, a correct sender's rb-broadcast at time \(t\) is rb-delivered by all correct parties by \(\max(\text{GST},t)+d_s\delta\), or by \(\max(\text{GST},t)+d_o\delta\) under the optimistic fault bound. A first correct rb-delivery at time \(t\) propagates to all correct parties by \(\max(\text{GST},t)+d_t\delta\). These bounds also cover broadcasts and deliveries before GST.

\begin{table}[h]
    \centering
    \caption{Latency Parameters of the RBC Protocols.}%
    \label{tab:broadcast_comparison}
    \newcommand{\hd}[2]{\textbf{\makecell{#1\\#2}}}
    \scriptsize
    \setlength{\tabcolsep}{2pt}
    \begin{threeparttable}
    \begin{tabular}{@{}l c c c c@{}}
    \toprule
        \hd{Broadcast}{protocol}
      & \hd{Opt.}{resilience ($f_o$)}
      & \hd{Opt.}{delay ($d_o$)}
      & \hd{Worst}{delay ($d_s$)}
      & \hd{Totality}{delay ($d_t$)} \\ [0.5ex]
    \midrule
      Bracha-RBC~\cite{bracha1987asynchronous}     & $f_o = f$                       & 3 & 3 & 2 \\
      Opt-RBC~\cite{shrestha2025optimistic} & $f_o = \floor{n/2} - f$\tnote{*} & 2 & 4 & 4 \\
    \bottomrule
    \end{tabular}%
    \begin{tablenotes}
    \scriptsize
    \item[*] \(f_o\approx n/6\) when \(f\approx n/3\).
    \end{tablenotes}
    \end{threeparttable}
\end{table}

\subsection{Erasure-coded RBC for Long Messages}

Following standard techniques from prior work of Cachin and Tessaro~\cite{cachin2005asynchronous} and Shrestha et al.~\cite{shrestha2025optimistic}, both Bracha-RBC and Opt-RBC can be made communication-efficient for long messages, without increasing latency, using Reed-Solomon erasure coding. %
For an \(L\)-bit input, the resulting protocols have communication complexity \(O(nL+\kappa n^2\log n)\) per broadcast, where \(\kappa\) denotes the hash output length, and they have balanced communication, meaning that all parties send and receive approximately the same number of bits.

In consensus protocols, having leaders use erasure-coded RBC to disseminate their blocks avoids the leader-bottleneck problem~\cite{yang2022dispersedledger,danezis2022narwhal} and, compared to using a separate mempool layer, does not increase the good-case commit latency.
The flip side is that total communication is inversely proportional to the rate of the erasure code used, and that encoding and decoding require CPU time.
Both erasure-coded RBC protocols use Reed-Solomon codes parameterized to recover the data from \(\ceil{(n-f+1)/2}\) out of \(n\) fragments (this is \(f+1\) when \(n=3f+1\)); with \(n\approx 3f\), this gives a rate of approximately \(1/3\), i.e., a \(3\times\) communication overhead due to erasure coding.

It is also worth noting that the techniques of Cachin and Tessaro and Shrestha et al. rely only on authenticated channels and cryptographic hash functions like SHA256, which are post-quantum-secure.

%% file: protocol.tex
\section{The \name Protocol}%
\label{sec: protocol}

\subsection{High-Level Structure}%
\label{sec:high-level-protocol}

\name implements total-order broadcast as defined in~\Cref{def: tob}.
In a nutshell, parties tob-submit blocks that \name must tob-deliver in the same order at all parties; moreover, if every correct party keeps tob-submitting blocks, then every correct party must tob-deliver infinitely many blocks tob-submitted by correct parties.

To order tob-submitted blocks, \name executes an infinite sequence of rounds $1, 2, \dots$.
In each round \(r\), a predetermined leader uses reliable broadcast to try to assign a block \(b\) and a parent \(r'\) to round \(r\), where \(r'\) is a round \(r' < r\) or \(r' =0\), and then parties try to agree on one of two mutually exclusive outcomes: to commit round \(r\) (only if the leader's RBC succeeded) or to disable round~\(r\).

The parent relation forms a tree of rounds. To ensure that committed rounds lie on a single chain starting at~\(0\), parties vote to commit a round only after determining that its proposal is safe.
At a party, round \(r'\) is a \emph{safe parent for round \(r\)}, written \(\fn{SafeParent}(r,r')\), when (recursively):
\begin{enumerate}
    \item \(0\leq r'<r\),
    \item \(r'=0\) or the party has a safe proposal for round \(r'\), and
    \item all rounds strictly between \(r'\) and \(r\) are disabled.
\end{enumerate}
A party has a \emph{safe proposal for round \(r\)} when it has rb-delivered that round's proposal and determined that its parent is a safe parent for \(r\).
The local flag \(\var{safe}[r]\) records this fact; \(\var{safe}[0]\) is true initially.
We use ``round \(r\) is safe'' as shorthand for \(\var{safe}[r]=\True\) at the party in question.

Finally, once a party has committed a round and has its safe proposal, it can tob-deliver the previously undelivered blocks, in order, on the chain from \(0\) to that round.
An example appears in~\Cref{fig:example-1}.
Note that the high-level structure of \name, as just described, is similar to that of the Simplex protocol~\cite{chanSimplexConsensusSimple2023}.

\input{figure_example_1}

\input{figure_lionfish_protocol_prose}

\subsection{The Concrete \name Protocol}

\Cref{fig:lionfish-protocol-prose} presents the concrete \name protocol in detail.
It specifies the local state variables of a party and the atomic steps that the party may take
(whenever several steps are enabled, the party may execute any of them).

Each round, parties restart their round timer (set for \(\Dto=(d_s+d_t)\Delta\) time) and a predetermined leader for the round proposes a block and a safe parent round (Steps \textbf{Enter round} and \textbf{Propose}).
\name uses a predetermined, fair leader schedule, so all parties know in advance who the leader of a round is and every party leads infinitely many rounds.

Once parties have rb-delivered the proposal for their current round and determined that its parent is a safe parent for that round, and the round timer has not expired, they vote to commit their current round (Step \textbf{Vote}; note that including a block digest or parent is unnecessary because they are uniquely determined by RBC).
If a party's round timer expires before it has voted, it times out: it rn-raises the round's timeout flag using the reliable notification protocol depicted in~\Cref{fig:tbn} (Step \textbf{Timeout}) and forgoes voting to commit that round. Upon rn-confirming the round's timeout flag, it disables the round (Step \textbf{Disable}).
Finally, parties finish their current round and enter the next round (Step \textbf{Advance round}) as soon as they have a safe proposal and have either voted to commit or rn-raised their timeout flag, or as soon as the current round is disabled.

Importantly, note that voting to commit a round \(r\) and rn-raising the round-\(r\) timeout flag are mutually exclusive.
By quorum intersection, this guarantees that no round is ever both committed and disabled.
Note, however, that it is possible for a party to rb-deliver the leader's proposal for a round \(r\) after enough parties have rn-raised the round-\(r\) timeout flag for it to be eventually marked as disabled.
Thus, a disabled round can still be safe, and thus serve as a parent of a future round  (e.g., round 3 in~\Cref{fig:example-1} is disabled but safe and serves as parent of round 5).

The rest of the protocol machinery is agnostic to the current round: parties keep participating in the reliable broadcast and reliable notification subprotocols for all rounds, and they keep tracking what proposals they rb-deliver in which round (Step \textbf{RB-deliver}), which proposals they know are safe (Step \textbf{Mark safe}), and which rounds they have disabled (Step \textbf{Disable}).
This is important for liveness, as advancing rounds (Step \textbf{Advance round}) and making a proposal (Step \textbf{Propose}) depend on being able to assign previous rounds a safe proposal or to disable them, which in turn depend on reliable broadcast or reliable notification making progress for those rounds.

Finally, a party commits round \(r\) upon receiving \(n-f\) votes to commit it (Step \textbf{Commit}). Once it also has a safe proposal for \(r\), it tob-delivers and appends to \(\var{delivered}\), in order, every non-\(\bot\) block in \(\fn{Log}(r)\) that it has not previously tob-delivered (Step \textbf{Deliver}).

Note that the protocol features a two-stage pipeline: in the good case, Step \textbf{Vote} sends votes to commit round \(r\), and Steps \textbf{Advance round} and \textbf{Enter round} move parties to round \(r+1\) without waiting for round \(r\) to commit. The next leader then rb-broadcasts its proposal through Step \textbf{Propose}. Thus, the commit phase of round \(r\) overlaps with the proposal phase of round \(r+1\).
We will later present another version of the protocol (\Cref{fig:lionfish-protocol-speculative-pipelining}) that uses speculation to achieve more aggressive pipelining.

\subsection{Correctness Sketch}

As explained in~\Cref{sec:high-level-protocol}, parties vote to commit only rounds for which they have a safe proposal. Because a committed round is never disabled, the safe-parent condition ensures that all committed rounds lie on the same chain of parents, so tob-delivering the blocks on that chain establishes the Total Order property of TOB.

Two key design elements guarantee liveness.
First, in each round, note that either all correct parties rb-deliver the leader's proposal or they all rn-raise their round-\(r\) timeout flag (or both).
Thus, every round is eventually resolved as having a safe proposal or as disabled, and correct parties never get stuck in a round.

Second, setting the round timer to a duration of \(\Dto=(d_s+d_t)\Delta\) time ensures that every round with a correct leader first entered by a correct party at or after GST commits: in the worst case, the leader enters the round late by at most \(\max(d_t,2)=d_t\) (the maximum of the totality delay of the reliable broadcast used and the totality delay of reliable notification), and then its proposal is rb-delivered in at most the worst-case delay \(d_s\) of the reliable broadcast used; thus, if \(\Dto=(d_s+d_t)\Delta\), all correct parties vote to commit the round before they time out.
For Bracha-RBC, we get \(\Dto=5\Delta\), and for Opt-RBC we get \(\Dto=8\Delta\).
Since correct leaders occur infinitely often and every correct party keeps submitting blocks, their committed chains deliver infinitely many correct submissions, establishing the liveness property of TOB.

Finally, \name ensures Totality because all correct parties commit every round led by a correct leader and first entered by a correct party at or after GST, obtain its safe proposal, and tob-deliver the blocks of that round and its ancestors. A correct party that has not yet tob-delivered a block therefore tob-delivers it through Step \textbf{Deliver} for a later round led by a correct leader.

\subsection{Optimistic and Speculative Variants}
\label{sec: variants}

We propose four variants of the protocol: protocols \(\shortname\), \(\shortname_{opt}\), \(\shortname^s\), and \(\shortname_{opt}^s\); each of them corresponds to a choice of RBC algorithm (Bracha-RBC or Opt-RBC, indicated by the \(opt\) subscript) and to a proposal rule (normal proposal or speculative proposal, indicated by the \(s\) superscript).

Protocols~\(\shortname\) and \(\shortname_{opt}\) are both instances of the protocol described in~\Cref{fig:lionfish-protocol-prose} with different RBC implementations.
Protocol \(\shortname\) uses Bracha-RBC~\cite{bracha1987asynchronous} and has good-case commit latency at most \(4\delta\) (\Cref{lemma: good-case-commit-latency}): \(3\delta\) for Bracha-RBC and one further message delay for all correct parties to tob-deliver the proposed block.
Its optimistic amortized block time is at most \(3\delta\): under the optimistic conditions, each successive leader can propose within one Bracha-RBC delay of the previous proposal (\Cref{corollary: optimistic-block-time-delivery}).

Protocol \(\shortname_{opt}\) uses Opt-RBC~\cite{shrestha2025optimistic} and has good-case commit latency at most \(5\delta\). Under the optimistic conditions of~\Cref{def: optimistic-commit-latency}, including correct previous leaders and at most \(f_o=\floor{n/2}-f\) Byzantine parties, its optimistic commit latency is at most \(3\delta\). The fault bound guarantees at least \(n-f_o-1=\ceil{n/2}+f-1\) correct nonleaders, enough for the optimistic RBC path; it leaves the protocol's \(n-f\) commit quorum unchanged.
The optimistic amortized block time of \(\shortname_{opt}\) is at most \(2\delta\), the optimistic delay of Opt-RBC (\Cref{corollary: optimistic-block-time-delivery}).

Protocols \(\shortname^s\) and \(\shortname_{opt}^s\) are variants of \(\shortname\) and \(\shortname_{opt}\), respectively, where leaders can speculatively propose before they determine that the previous round's proposal is safe.
Both follow the description in~\Cref{fig:lionfish-protocol-speculative-pipelining}, with RBC instantiated with Bracha-RBC and with Opt-RBC, respectively. Their optimistic commit latencies are at most \(4\delta\) and \(3\delta\) under~\Cref{def: optimistic-commit-latency}, as established in~\Cref{lemma: optimistic-latency}.
In~\Cref{fig:lionfish-protocol-speculative-pipelining}, the differences from the non-speculative variants (\Cref{fig:lionfish-protocol-prose}) are highlighted in {\color{blue}blue}.

The main difference compared to the non-speculative protocol is that we add Step \textbf{Speculative propose}, in which the party \(p_i\) makes a speculative proposal for a round \(r\) greater than its current round.
This step is enabled for a round \(r>\var{curr\_round}\) as soon as (a) \(r-k>0\) and \(p_i\) has marked round \(r-k\) safe, where \(k\) is the proposal-pipeline depth parameter of the protocol, and (b) \(p_i\) has received, but not necessarily rb-delivered, the proposal for round \(r-1\) from the leader of \(r-1\).
The party remembers this receipt, so the step can become enabled when the safety condition holds later.
Party \(p_i\) then picks a block \(b\) and proposes it for round \(r\) with parent \(r-1\) by rb-broadcasting \(\tuple{p_i, r, \tuple{r-1,b}}\).
To avoid proposing multiple times in the same round, we also add a boolean variable \(\var{proposed}[r]\) tracking whether \(p_i\) has already proposed in round \(r\), and we do not propose (speculatively or otherwise) if this variable is already set.
Finally, to accelerate recovery from failed speculation, we modify Step \textbf{Disable} so that, upon rn-confirming the round-\(r\) timeout flag, a party whose own round-\(r\) timer expired immediately rn-raises the timeout flags of the rounds \(r'\) with \(\max(\var{curr\_round},r+1) \leq r' \leq r+k-1\) in which it has neither voted nor already aborted, recording each rn-raise in a new flag \(\var{aborted}[r']\) and forgoing voting in those rounds (the \(\var{voted}\) and \(\var{timed\_out}\) variables accordingly become per-round maps).
These rn-raises run in parallel and can let parties skip speculative rounds that are bound to fail without waiting for their timers.

Step \textbf{Speculative propose} overlaps proposal RBCs while requiring a safe anchor \(k\) rounds back.
Under the optimistic conditions, Bracha-RBC completes within \(3\delta\) and Opt-RBC within \(2\delta\). With \(k=3\) and \(k=2\), respectively, both speculative variants have optimistic amortized block time at most \(\delta\) (\Cref{def: optimistic-amortized-block-time,lemma: optimistic-block-time-pipelined}).
This bounds the long-run average time per tob-delivered block. Individual proposal intervals can be longer when fast proposals reach the safety guard before the older RBC completes.

Of course, while speculative proposals help shorten block time in good cases, they weaken bad-case progress guarantees.
A speculative proposal may fail to be tob-delivered because an earlier proposal it builds on failed to become safe, and this is an inherent tradeoff.
A sufficient condition for tob-delivery of the last round's block is a window of \(2k-1\) consecutive rounds with correct leaders whose first round is first entered by a correct party at or after GST (\Cref{lemma: liveness-gst-pipeline}).
The leaders need not be distinct: the same correct party may lead all \(2k-1\) rounds.

Also note that votes to commit do not participate in the proposal pipeline: Step \textbf{Vote} still requires having determined the current round safe, which implies having rb-delivered the current round's proposal and all its ancestors.
Safety then rests, as in the non-speculative protocol, on the fact that a correct party never both votes to commit a round and rn-raises its timeout flag: Step \textbf{Vote} is disabled by \(\var{timed\_out}\) and \(\var{aborted}\), while Steps \textbf{Timeout} and \textbf{Disable} rn-raise a round's flag only if the party has not voted in that round. Each rn-raise sets \(\var{timed\_out}\) or \(\var{aborted}\), preventing any subsequent vote in that round.
By quorum intersection, no round is ever both committed and disabled, exactly as in the non-speculative protocol.
Finally, note that Step \textbf{Disable} only rn-raises in rounds up to \(r+k-1\) if the party really timed out in round \(r\).
An abort confirmation alone therefore does not create a cascade of aborts beyond its speculative window.

\input{figure_lionfish_pipelined_protocol_prose}
\input{figure_btn}

%% file: figure_example_1.tex
\begin{figure}[t]
  \centering
  \begin{tikzpicture}[
      x=1.4cm,
      block/.style={draw, rounded corners=2pt, minimum width=0.95cm,
        minimum height=0.7cm, font=\small, fill=blue!8},
      committed/.style={block, fill=blue!15, very thick},
      disabled/.style={draw, dashed, rounded corners=2pt, minimum width=0.95cm,
        minimum height=0.7cm, font=\small\itshape, text=gray, fill=gray!8},
      parent/.style={-{Stealth[length=5pt]}, thick},
      rlabel/.style={font=\scriptsize\sffamily, gray},
    ]
    \foreach \r in {0,1,2,3,4,5} {
      \node[rlabel] at (\r, -0.7) {$r{=}\r$};
    }

    \node[block, fill=white, very thick] (b0) at (0,0) {$\bot$}; %
    \node[committed] (b1) at (1,0) {$v$};
    \node[disabled]  (b2) at (2,0) {};
    \node[disabled, text=black]  (b3) at (3,0) {$v'$};
    \node[disabled, text=black]  (b4) at (4,0) {$w$};
    \node[committed] (b5) at (5,0) {$v''$};

    \draw[parent] (b1) -- (b0);
    \draw[parent] (b3) to[bend left=28] (b1); %
    \draw[parent] (b5) to[bend left=28] (b3); %
    \draw[parent, dashed, thin] (b4) to[bend right=32] (b1); %
  \end{tikzpicture}
  \caption{An example \name execution. Solid blue boxes are committed rounds
($1,5$); dashed gray boxes are disabled rounds ($2,3,4$); $\bot$ is the genesis.
Arrows point from each round to its assigned parent, if it has one.
Round~$1$ is committed and has block \(v\) and parent \(0\).
Round~$2$ has no block or parent and is disabled.
Round~$3$ is disabled too, but it has an assigned block~$v'$ and parent round \(1\);
it skips the disabled round~$2$.
Round~$4$ is also disabled but has block~$w$ and parent round \(1\); it forks off of the
tob-delivered chain, so its block will never be tob-delivered.
Round~$5$ has parent \(3\) even though \(3\) is disabled; that is fine since round \(3\) has a safe proposal and round \(4\) is disabled.
Each parent is safe for the round that points to it.
The tob-delivered total order is given by the chain of rounds~$1,3,5$, namely $v, v', v''$.}%
  \label{fig:example-1}
\end{figure}

%% file: figure_lionfish_protocol_prose.tex
\begin{figure*}[!ht]
\small
    \begin{boxedminipage}[t]{\textwidth}
    \textbf{State variables:}
    \begin{itemize}
        \item $\var{curr\_round}$: the current round, initialized to 1.
        \item $\var{submitted}$: the set of blocks tob-submitted for ordering, initially empty.
        \item $\var{delivered}$: the sequence of blocks tob-delivered so far, initially empty.
        \item $\var{proposal}[r]$: the proposal $\tuple{r', b}$ rb-delivered for round \(r\) from the leader of round $r$, or \(\bot\) if none has been rb-delivered so far; $r'$ is the parent round and $b$ is the block; initially, $\var{proposal}[0] = \tuple{\bot, \text{genesis}}$ and $\var{proposal}[r] = \bot$ for $r > 0$.
        \item $\var{safe}[r]$: a boolean map indicating whether round $r$ has a safe proposal; $\var{safe}[0] = \True$ and $\var{safe}[r] = \False$ for $r > 0$. \(\var{safe}[r]\) implies \(\var{proposal}[r]\neq \bot\).
        \item $\var{disabled}[r]$: a boolean map indicating whether the round-$r$ timeout flag was rn-confirmed (so round $r$ is disabled), initially everywhere $\False$.
        \item $\var{committed}[r]$: a boolean map indicating whether the round $r$ was committed, initially \(\False\) everywhere.
        \item $\var{voted}$: boolean variable tracking whether the party has voted to commit the current round, reset to $\False$ each round.
        \item $\var{timed\_out}$: boolean variable tracking whether the party has rn-raised its timeout flag for the current round, reset to $\False$ each round.
    \end{itemize}

    \textbf{Helpers:}
    \begin{itemize}
        \item $\fn{SafeParent}(r, r')$:
            Return $\True$ when $0 \le r' < r$, $\,\var{safe}[r'] = \True$, and $\var{disabled}[r''] = \True$ for every round $r' < r'' < r$.
        \item $\fn{Log}(r)$:
            The sequence of blocks of the chain ending at round $r$, defined recursively: $\fn{Log}(0)$ is the empty sequence, and for $r > 0$ with $\var{proposal}[r] = \tuple{r', b}$ and \(\var{safe}[r]\), $\fn{Log}(r)$ is $\fn{Log}(r')$ followed by $b$.
    \end{itemize}

    \textbf{Protocol steps:}
    \begin{itemize}

        \item \textbf{Init.} Upon startup, enter round $1$.

        \item \textbf{Submit.} Upon a request to tob-submit a block $b$, add $b$ to $\var{submitted}$.

        \item \textbf{Enter round.} Upon entering a round $r$, set $\var{curr\_round}\gets r$, reset the round timer to expire in \(\Dto=(d_s+d_t)\Delta\) time, set $\var{voted}\gets\False$ and $\var{timed\_out}\gets\False$, and, if $\party{i}$ is the leader of round $r$ ($L_r = p_i$), propose.

        \item \textbf{Propose.}\label{step:lionfish-prose-propose} Upon proposing, determine the highest round $r$ such that $\fn{SafeParent}(\var{curr\_round}, r)$ and pick a block $b \in \var{submitted} \setminus \var{delivered}$ where \(b\notin \fn{Log}(r)\) (or $b = \bot$ if no such block exists), and rb-broadcast the message $\tuple{p_i, \var{curr\_round}, \tuple{r, b}}$.

        \item \textbf{RB-deliver.}\label{step:lionfish-prose-deliver} Upon \Call{rb\_deliver}{} $L_r$'s round-$r$ proposal $\tuple{L_r, r, \tuple{r', b}}$ for some round \(r\), set $\var{proposal}[r]\gets \tuple{r', b}$.

        \item \textbf{Mark safe.} Upon $\var{proposal}[r] = \tuple{r', b}$ and $\fn{SafeParent}(r, r')$ for some round $r$, set $\var{safe}[r]\gets\True$.

        \item \textbf{Vote.}\label{step:lionfish-prose-vote} Upon $\var{safe}[\var{curr\_round}] = \True$, $\var{timed\_out} = \False$, and $\var{voted} = \False$, send $\tuple{\Commit, \var{curr\_round}}$ to all parties and set $\var{voted}\gets\True$.

        \item \textbf{Timeout.} Upon the round timer expiring, if $\var{voted} = \False$ then trigger \Call{rn\_raise}{$\tuple{\text{timeout}, \var{curr\_round}}$} and set $\var{timed\_out}\gets\True$.

        \item \textbf{Disable.}\label{step:lionfish-prose-disable} Upon \Call{rn\_confirm}{$\tuple{\text{timeout},r}$} for some round $r$, set $\var{disabled}[r]\gets\True$.

        \item \textbf{Commit.} Upon receiving $\tuple{\Commit, r}$ from $n - f$ parties for some round \(r\), set \(\var{committed}[r]\gets \True\).

        \item \textbf{Deliver.} Upon \(\var{committed}[r]=\True\), for some round \(r\), and \(\var{safe}[r]=\True\), tob-deliver and append to \(\var{delivered}\), in order, every non-\(\bot\) block of $\fn{Log}(r)$ not previously tob-delivered.

        \item \textbf{Advance round.} Upon $\var{safe}[\var{curr\_round}] = \True$ and either \(\var{voted}=\True\) or \(\var{timed\_out}=\True\), or upon $\var{disabled}[\var{curr\_round}] = \True$, enter round $\var{curr\_round}+1$.
    \end{itemize}

    \end{boxedminipage}
    \caption{The \name protocol; code for party \party{i}. A reliable-notification implementation appears in~\Cref{fig:tbn}.}%
    \label{fig:lionfish-protocol-prose}
\end{figure*}

%% file: figure_lionfish_pipelined_protocol_prose.tex
\begin{figure*}[!ht]
\small
    \begin{boxedminipage}[t]{\textwidth}
    \textbf{Parameters:}
    \begin{itemize}
        \item {\color{blue}$\var{k}$: the proposal pipeline depth; \(\var{k}=3\) if using Bracha-RBC and \(\var{k}=2\) if using Opt-RBC.}
    \end{itemize}
    \textbf{State variables:}
    \begin{itemize}
        \item $\var{curr\_round}$: the current round, initialized to 1.
        \item $\var{submitted}$: the set of blocks tob-submitted for ordering, initially empty.
        \item $\var{delivered}$: the sequence of blocks tob-delivered so far, initially empty.
        \item $\var{proposal}[r]$: the proposal $\tuple{r', b}$ rb-delivered for round \(r\) from the leader of round $r$, or \(\bot\) if none has been rb-delivered so far; $r'$ is the parent round and $b$ is the block; initially, $\var{proposal}[0] = \tuple{\bot, \text{genesis}}$ and $\var{proposal}[r] = \bot$ for $r > 0$.
        \item $\var{safe}[r]$: a boolean map indicating whether round $r$ has a safe proposal; $\var{safe}[0] = \True$ and $\var{safe}[r] = \False$ for $r > 0$. \(\var{safe}[r]\) implies \(\var{proposal}[r]\neq \bot\).
        \item $\var{disabled}[r]$: a boolean map indicating whether the round-$r$ timeout flag was rn-confirmed (so round $r$ is disabled), initially everywhere $\False$.
        \item $\var{committed}[r]$: a boolean map indicating whether the round $r$ was committed, initially \(\False\) everywhere.
        \item {\color{blue}$\var{voted}[r]$: a boolean map tracking whether the party has voted to commit round $r$, initially everywhere $\False$.}
        \item {\color{blue}$\var{timed\_out}[r]$: a boolean map tracking whether the party's round-$r$ timer expired and it rn-raised the round-$r$ timeout flag (Step \textbf{Timeout}), initially everywhere $\False$.}
        \item {\color{blue}$\var{aborted}[r]$: a boolean map tracking whether the party rn-raised the round-$r$ timeout flag to abort speculation in round $r$, without waiting for its round-$r$ timer (Step \textbf{Disable}), initially everywhere $\False$.}
        \item {\color{blue}$\var{proposed}[r]$: a boolean map indicating whether this party has already rb-broadcast a proposal for round $r$, initially everywhere $\False$.}

    \end{itemize}

    \textbf{Helpers:}
    \begin{itemize}
        \item $\fn{SafeParent}(r, r')$:
            Return $\True$ when $0 \le r' < r$, $\,\var{safe}[r'] = \True$, and $\var{disabled}[r''] = \True$ for every round $r' < r'' < r$.
        \item $\fn{Log}(r)$:
            The sequence of blocks of the chain ending at round $r$, defined recursively: $\fn{Log}(0)$ is the empty sequence, and for $r > 0$ with $\var{proposal}[r] = \tuple{r', b}$ and \(\var{safe}[r]\), $\fn{Log}(r)$ is $\fn{Log}(r')$ followed by $b$.
    \end{itemize}

    \textbf{Protocol steps:}
    \begin{itemize}

        \item \textbf{Init.} Upon startup, enter round $1$.

        \item \textbf{Submit.} Upon a request to tob-submit a block $b$, add $b$ to $\var{submitted}$.

        \item \textbf{Enter round.} Upon entering a round $r$, set $\var{curr\_round}\gets r$, reset the round timer to expire in \(\Dto=(d_s+d_t)\Delta\) time, and, if $\party{i}$ is the leader of round $r$ ($L_r = p_i$), propose.

        \item \textbf{Propose.}\label{step:lionfish-pipelined-prose-propose} Upon proposing, {\color{blue} if $\var{proposed}[\var{curr\_round}] = \False$, } determine the highest round $r$ such that $\fn{SafeParent}(\var{curr\_round}, r)$ and pick a block $b \in \var{submitted} \setminus \var{delivered}$ such that \(b\notin \fn{Log}(r)\) (or $b = \bot$ if no such block exists), and rb-broadcast the message $\tuple{p_i, \var{curr\_round}, \tuple{r, b}}$. {\color{blue} Set $\var{proposed}[\var{curr\_round}]\gets\True$.}

        \item {\color{blue}\textbf{Speculative propose.} When, for some \(r>\var{curr\_round}\), $\party{i}$ has received $\tuple{L_{r-1}, r-1, \tuple{r', b'}}$ directly from $L_{r-1}$ (the first message of its round-$(r{-}1)$ RBC instance), $L_r=p_i$, $\var{proposed}[r]=\False$, \(r-k>0\), and \(\var{safe}[r-k]=\True\), rb-broadcast a round-$r$ proposal $\tuple{p_i, r, \tuple{r-1, b}}$ for some block $b \in \var{submitted}\setminus\var{delivered}$ (or $b=\bot$ if no such \(b\) exists), and set $\var{proposed}[r]\gets\True$.}

        \item \textbf{RB-deliver.}\label{step:lionfish-pipelined-prose-deliver} Upon \Call{rb\_deliver}{} $L_r$'s round-$r$ proposal $\tuple{L_r, r, \tuple{r', b}}$ for some round \(r\), set $\var{proposal}[r]\gets \tuple{r', b}$.

        \item \textbf{Mark safe.} Upon $\var{proposal}[r] = \tuple{r', b}$ and $\fn{SafeParent}(r, r')$ for some round $r$, set $\var{safe}[r]\gets\True$.

        \item \textbf{Vote.}\label{step:lionfish-pipelined-prose-vote} Upon $\var{safe}[\var{curr\_round}] = \True$, $\var{timed\_out}[\var{curr\_round}] = \False$, {\color{blue}$\var{aborted}[\var{curr\_round}] = \False$,} and $\var{voted}[\var{curr\_round}] = \False$, send $\tuple{\Commit, \var{curr\_round}}$ to all parties and set $\var{voted}[\var{curr\_round}]\gets\True$.

        \item \textbf{Timeout.} Upon the round timer expiring, if $\var{voted}[\var{curr\_round}] = \False$ then trigger \Call{rn\_raise}{$\tuple{\text{timeout}, \var{curr\_round}}$} and set $\var{timed\_out}[\var{curr\_round}]\gets\True$.

        \item \textbf{Disable.}\label{step:lionfish-pipelined-prose-disable} Upon \Call{rn\_confirm}{$\tuple{\text{timeout},r}$} for some round $r$, set \(\var{disabled}[r]\gets \True\). {\color{blue}Moreover, if $\var{timed\_out}[r] = \True$, then, for every round $r'$ with $\max(\var{curr\_round},r+1) \le r' \le r+k-1$, $\var{voted}[r'] = \False$, and $\var{aborted}[r'] = \False$, trigger \Call{rn\_raise}{$\tuple{\text{timeout}, r'}$} and set $\var{aborted}[r']\gets\True$.}

        \item \textbf{Commit.} Upon receiving $\tuple{\Commit, r}$ from $n - f$ parties for some round \(r\), set \(\var{committed}[r]\gets \True\).

        \item \textbf{Deliver.} Upon \(\var{committed}[r]=\True\), for some round \(r\), and \(\var{safe}[r]=\True\), tob-deliver and append to \(\var{delivered}\), in order, every non-\(\bot\) block of $\fn{Log}(r)$ not previously tob-delivered.

        \item \textbf{Advance round.} Upon $\var{safe}[\var{curr\_round}] = \True$ and either {\color{blue} \(\var{voted}[\var{curr\_round}]=\True\) or \(\var{timed\_out}[\var{curr\_round}]=\True\)  or \(\var{aborted}[\var{curr\_round}]=\True\)}, or upon $\var{disabled}[\var{curr\_round}] = \True$, enter round $\var{curr\_round}+1$.
    \end{itemize}

    \end{boxedminipage}
    \caption{The \name protocol with speculative pipelining; code for party \party{i}. The differences with~\Cref{fig:lionfish-protocol-prose} are highlighted in {\color{blue}blue}. A reliable-notification implementation appears in~\Cref{fig:tbn}.}%
    \label{fig:lionfish-protocol-speculative-pipelining}
\end{figure*}

%% file: figure_btn.tex
\begin{figure}[ht]
    \small
    \begin{boxedminipage}[t]{\columnwidth}
        \begin{enumerate}[leftmargin=*]
            \item \textbf{Vote.} Upon invoking \Call{rn\_raise}{$e$}, \party{i} broadcasts a $\sig{\Vote, e}$ message.

            \item \textbf{Accept.} Upon receiving $\sig{\Vote, e}$ messages from $n-f$ parties, \party{i} broadcasts a $\sig{\Accept, e}$ message.

            \item \textbf{Confirm.} Upon receiving $\sig{\Accept, e}$ messages from $2f+1$ parties, \party{i} rn-confirms \(e\) by invoking \Call{rn\_confirm}{$e$}.

            \item \textbf{Cascade.} Upon receiving $\sig{\Accept, e}$ messages from $f+1$ parties, and not having sent an $\sig{\Accept, e}$ message, \party{i} broadcasts an $\sig{\Accept, e}$ message.
        \end{enumerate}
    \end{boxedminipage}
\caption{Reliable notification protocol.}%
\label{fig:tbn}
\end{figure}

%% file: security_analysis2.tex
\section{Security Analysis}%
\label{sec: security}

Throughout this section we reason about the state model of the \name protocol of~\Cref{fig:lionfish-protocol-prose}.
We first prove safety and liveness, then analyze reliable notification, latency bounds, and speculative pipelining.

Recall that each round \(r\) has at most one assigned proposal: when a party rb-delivers the round-\(r\) proposal \(\tuple{L_r, r, \tuple{r', b}}\), it sets \(\var{proposal}[r] \gets \tuple{r', b}\), where \(r'\) is the parent and \(b\) the block assigned to round \(r\); by the Uniqueness and Totality properties of reliable broadcast, all correct parties that set \(\var{proposal}[r]\) agree on the same \(\tuple{r', b}\).
We say that round \(r\) is \emph{disabled} when its timeout flag is rn-confirmed (i.e., a party sets \(\var{disabled}[r] \gets \True\) upon \Call{rn\_confirm}{$\tuple{\text{timeout}, r}$}), and that a party \emph{commits round \(r\)} when it applies Step \textbf{Commit} after receiving \(n-f\) \(\tuple{\Commit, r}\) messages.

\begin{definition}[Parent round]
    A round \(r'\) is the \emph{parent} of a round \(r\) if and only if \(\var{proposal}[r] = \tuple{r', b}\) for some block \(b\) and \(\var{safe}[r] = \True\); that is, \(r'\) is the safe parent assigned to round \(r\).
    The genesis round \(0\) has no parent.
\end{definition}

\begin{definition}[Ancestor round]
    A round \(r'\) is an \emph{ancestor} of a round \(r\) when any of the following conditions hold:
    \begin{itemize}
        \item \(r' = r\);
        \item \(r'\) is the parent of round \(r\);
        \item \(r'\) is the parent of a round \(r''\) and \(r''\) is an ancestor of round \(r\).
    \end{itemize}
\end{definition}

\subsection{Safety}

\begin{lemma}
    \label{lemma: commit-not-disabled}
    If a correct party commits a round~\(r\), then no correct party ever disables round~\(r\).
\end{lemma}
\begin{proof}
    If a correct party \party{i} commits round \(r\), then \party{i} must have received \(n-f\) \(\tuple{\Commit, r}\) messages.
    Voting to commit round \(r\) and rn-raising the round-\(r\) timeout flag are mutually exclusive for a correct party: by Step \textbf{Vote}, a correct party sends \(\tuple{\Commit, r}\) only while \(\var{timed\_out} = \False\), and by Step \textbf{Timeout} it triggers \Call{rn\_raise}{$\tuple{\text{timeout}, r}$} only while \(\var{voted} = \False\) (after which it sets \(\var{timed\_out} = \True\)). Hence, in round \(r\), a correct party that has already rn-raised its timeout flag will not subsequently vote (the \textbf{Vote} guard fails), and one that has already voted will not subsequently rn-raise it (the \textbf{Timeout} guard fails); so it never does both.
    Now suppose toward a contradiction that some correct party disables round \(r\), i.e., rn-confirms the round-\(r\) timeout flag.
    By the Validity property of reliable notification (\Cref{def: reliable-notification}), at least \(n-2f\) correct parties called \Call{rn\_raise}{$\tuple{\text{timeout}, r}$}.
    Since \((n-f)+(n-2f) = n+(n-3f) > n\), one of these correct parties also sent a \(\tuple{\Commit, r}\) message to \party{i}, and thus both voted to commit round \(r\) and rn-raised its round-\(r\) timeout flag.
    But we have previously established that no correct party does both — a contradiction.
\end{proof}

\begin{lemma}
    \label{lemma: chain}
    If a correct party commits a round \(r\), then the genesis round \(0\) is an ancestor of it.
\end{lemma}
\begin{proof}
    For a correct party to commit round \(r\), the party must have received \(\tuple{\Commit, r}\) messages from \(n-f\) parties, of which at least \(f+1\) are correct.
    A correct party sends \(\tuple{\Commit, r}\) only once \(\var{safe}[r] = \True\), and a correct party sets \(\var{safe}[k] = \True\) only when \(\var{proposal}[k] = \tuple{k', b}\) and \(\fn{SafeParent}(k, k')\) hold; that is, only once it has marked the parent of round \(k\) safe.
    Hence every ancestor \(k > 0\) of round \(r\) has a parent that was marked safe by some correct party.
    Let \(r'\) be the ancestor of \(r\) for the smallest \(r' > 0\), and let \(s\) be the parent of round \(r'\).
    There are three possibilities for \(s\): either \(s > 0\), \(s = 0\), or \(s < 0\).
    The first case is not possible, since then \(s\) would be an ancestor of \(r\) for a round smaller than \(r'\).
    The case \(s < 0\) is not possible either, since \(\fn{SafeParent}(r', s)\) requires \(0 \le s\).
    Thus the only possible value for \(s\) is \(0\), which proves the lemma.
\end{proof}

\begin{lemma}
    \label{lemma: commited-is-ancestor}
    If a correct party commits a round \(r > 0\), then round \(r\) is an ancestor of every round \(r' > r\) for which some correct party sets \(\var{safe}[r'] = \True\).
\end{lemma}
\begin{proof}
    Suppose some correct party sets \(\var{safe}[r'] = \True\) for a round \(r' > r\) and, towards a contradiction, that round \(r\) is not an ancestor of \(r'\).
    By the reasoning of Lemma~\ref{lemma: chain}, every ancestor \(k > 0\) of \(r'\) has a parent that was marked safe by some correct party, and that parent is itself an ancestor of \(r'\) lying strictly below \(k\) (since \(\fn{SafeParent}\) requires a parent to be a strictly smaller round); thus the ancestors of \(r'\) form a descending chain from \(r'\) down to \(0\).
    Consider the set of ancestors of \(r'\) that exceed \(r\). It is non-empty, since \(r'\) is an ancestor of itself and \(r' > r\); let \(s\) be its smallest element, and let \(s'\) be the parent of \(s\) (which exists because \(s > r \ge 1\), so \(s > 0\)).
    Now \(s'\) is also an ancestor of \(r'\), and \(\fn{SafeParent}(s, s')\) forces \(s' < s\). By minimality of \(s\), no ancestor of \(r'\) lies in the interval \((r, s)\), so \(s' \le r\). Moreover \(s' \neq r\): since \(s'\) is an ancestor of \(r'\) but \(r\) is not, we have \(s' < r\). Combining, \(s' < r < s\).
    For some correct party to set \(\var{safe}[s] = \True\), the predicate \(\fn{SafeParent}(s, s')\) must hold, which requires \(\var{disabled}[k''] = \True\) for every round \(s' < k'' < s\), and in particular for round \(r\) (as \(s' < r < s\)).
    But by Lemma~\ref{lemma: commit-not-disabled}, no correct party ever disables round \(r\), so \(\var{disabled}[r] = \True\) never holds for a correct party.
    This contradicts \(\fn{SafeParent}(s, s')\), and the lemma follows.
\end{proof}

\begin{lemma}%
    \label{lemma: round-agreement}
    If two correct parties \party{i} and \party{j} commit rounds \(r\) and \(r'\) respectively, then either \(r\) is an ancestor of \(r'\), or \(r'\) is an ancestor of \(r\).
\end{lemma}
\begin{proof}
    Without loss of generality assume that \(r \le r'\); if \(r = r'\) the claim is immediate, so assume \(r < r'\).
    Since round \(r'\) is committed, at least \(f+1\) correct parties sent \(\tuple{\Commit, r'}\), and by Step \textbf{Vote} each did so only after setting \(\var{safe}[r'] = \True\); hence some correct party set \(\var{safe}[r'] = \True\).
    As \(r' > r\) and round \(r\) is committed, Lemma~\ref{lemma: commited-is-ancestor} shows that \(r\) is an ancestor of \(r'\).
\end{proof}

\begin{theorem}[Total Order]
    \label{theorem: total-order}
    If a correct party tob-delivers a block \(b\) before a block \(b'\), then no correct party tob-delivers \(b'\) before \(b\).
\end{theorem}
\begin{proof}
    By Lemma~\ref{lemma: round-agreement}, committed rounds are totally ordered by ancestry; by the definition of \(\fn{Log}\), their logs form a chain under the prefix order.
    Step \textbf{Deliver} applies only to committed, safe rounds and appends previously undelivered, non-\(\bot\) blocks in log order.
    Removing \(\bot\) and repeated blocks preserves prefixes, so any two correct parties' delivered sequences are comparable under the prefix order and agree on relative block order.
\end{proof}

\subsection{Liveness}

\begin{lemma}
    \label{lemma: deliver-or-disable}
    For each round \(r\), either all correct parties rb-deliver the leader \(L_r\)'s proposal, or all correct parties disable round \(r\), or both.
\end{lemma}
\begin{proof}
    There are two possible cases for any given round \(r\): either at least one correct party rb-delivers the leader's proposal, or no correct party ever does so.
    In the first case, by the Totality property of reliable broadcast, all correct parties will eventually rb-deliver the leader's proposal.
    Otherwise, no correct party ever sets \(\var{safe}[r] = \True\) or votes to commit round \(r\), so after their round timers expire all correct parties trigger \Call{rn\_raise}{$\tuple{\text{timeout}, r}$}. By the Unanimity property of reliable notification, since all correct parties rn-raise the round-\(r\) timeout flag, all correct parties eventually rn-confirm it and disable round \(r\).
\end{proof}

A consequence of the previous lemma is that all correct parties eventually satisfy Step \textbf{Advance round} and progress to round \(r + 1\).
\begin{corollary}
    \label{corollary: higher-rounds}
    All correct parties keep entering higher rounds.
\end{corollary}

\begin{lemma}
    \label{lemma: optimistic2}
    If a correct party enters a round \(r\) at time \(t\), then all correct parties enter round \(r\) by \(\max(\text{GST},t)+d_t\delta\).
\end{lemma}
\begin{proof}
    A party enters round \(r\) only after marking safe or disabling every preceding round. Consider all rb-deliveries and rn-confirmations through which this party established those facts, including the entire ancestry of each safe proposal. RBC and RN totality propagate all these facts in parallel. Every other correct party can then mark the same rounds safe or disabled and apply Steps \textbf{Vote} and \textbf{Advance round} through the preceding rounds (a party that already timed out can advance once the round is safe).
    By~\Cref{lemma: 2Delta-btn} and~\Cref{tab:broadcast_comparison}, the totality delay of reliable broadcast (\(d_t = 2\) or \(d_t = 4\) message delays, depending on the protocol used) is always greater than or equal to the totality delay of the reliable notification protocol (\(2\) message delays).
    Thus, we conclude that all correct parties enter round \(r\) by time \(\max(\text{GST}, t) + d_t\delta\).
\end{proof}

\begin{lemma}
    \label{lemma: deliver-information}
    If a correct party rb-broadcasts a normal proposal for round \(r\) at time \(t\) whose parent is round \(r'\), then all correct parties mark round \(r'\) safe and disable all rounds \(r''\) such that \(r'<r''<r\) by time \(\max(\text{GST},t)+d_t\delta\).
\end{lemma}
\begin{proof}
    Step \textbf{Propose} requires \(\fn{SafeParent}(r,r')\), so the leader already knows that round \(r'\) is safe and all intervening rounds are disabled. Recursively, it has rb-delivered every non-genesis proposal in the ancestry of \(r'\) and rn-confirmed every timeout flag needed for its parent links. All these events occurred by \(t\). RBC totality and RN totality (\Cref{lemma: 2Delta-btn}), with \(d_t\geq 2\), propagate all those proposals and flags to every correct party by \(\max(\text{GST},t)+d_t\delta\). Starting from the initially safe genesis round, each party can therefore mark the whole chain through \(r'\) safe and disable the intervening rounds by that deadline.
\end{proof}

\begin{lemma}
    \label{lemma: liveness-gst}
    If the first correct party to enter round \(r\) does so at time \(t\) such that \(t\geq \text{GST}\), the leader \(L_r\) is correct, and the round timer duration is larger than \((d_t+d_s)\delta\), then every correct party will commit round \(r\).
\end{lemma}
\begin{proof}
    By Lemma~\ref{lemma: optimistic2}, all correct parties enter round \(r\) by time \(t+d_t\delta\).

    Let \(t_r\leq t+d_t\delta\) be the time at which \(L_r\) enters round \(r\) and rb-broadcasts its proposal. All correct parties rb-deliver that proposal by \(t_r+d_s\delta\). By~\Cref{lemma: deliver-information}, they also establish its safe parent and all required disabled rounds by \(t_r+d_t\delta\). Hence every correct party marks round \(r\) safe by \(t_r+\max(d_s,d_t)\delta\). For both RBC implementations, \(d_s\geq d_t\), so this is at most \(t+(d_t+d_s)\delta\).

    If the round timer duration is strictly larger than \((d_t+d_s)\delta\), no correct party will have timed out by \(t+(d_t+d_s)\delta\). Therefore, by this point, all correct parties will have marked round \(r\) safe, sent a \(\tuple{\Commit, r}\) message, and advanced to round \(r+1\). Since \(\delta\) is not known, the round timer duration has to be specified using the known upper bound on the message delay, \(\Delta\), which results in a round timer duration of \((d_t+d_s)\Delta\).

    These commit messages will experience at most one message delay. Thus, by time \(t+d_t\delta+d_s\delta+\delta\), every correct party will have received at least \(n-f\) \(\tuple{\Commit, r}\) messages, which triggers Step \textbf{Commit} causing every correct party to commit round \(r\) at this time, as long as the round timer duration is strictly larger than \((d_t+d_s)\delta\).
\end{proof}

\begin{theorem}[Liveness]
    \label{theorem: liveness}
    With a round timer duration larger than \((d_t+d_s)\delta\), and a leader schedule that ensures that correct leaders are chosen infinitely often, the protocol guarantees that if every correct party tob-submits infinitely many blocks, then every correct party tob-delivers infinitely many blocks tob-submitted by correct parties.
\end{theorem}
\begin{proof}
    Suppose every correct party tob-submits infinitely many blocks. By Corollary~\ref{corollary: higher-rounds} and the leader-schedule assumption, infinitely many rounds with correct leaders are first entered after GST. By~\Cref{lemma: liveness-gst}, every correct party commits each such round, obtains its safe proposal, and tob-delivers its block and those of its ancestors through Step \textbf{Deliver}.

    Since there are finitely many parties, some correct party leads infinitely many of these rounds. This party tob-submits infinitely many blocks, so the block-selection rule ensures that these committed chains contain infinitely many of its submissions.\footnote{Empty proposals contribute no block of their own; pending blocks already in a proposal's parent chain are delivered when that chain commits.} Every correct party therefore tob-delivers infinitely many blocks tob-submitted by correct parties.
\end{proof}

\begin{theorem}[Totality]
    \label{theorem: totality}
    If a correct party tob-delivers a block~\(b\), then every correct party eventually tob-delivers \(b\).
\end{theorem}
\begin{proof}
    A correct party tob-delivers \(b\) only through Step \textbf{Deliver} for a committed, safe round \(k\) with \(b\in\fn{Log}(k)\).
    By Corollary~\ref{corollary: higher-rounds} and the leader schedule, some round \(k'\geq k\) with a correct leader is first entered after GST. By Lemma~\ref{lemma: liveness-gst}, every correct party commits \(k'\) and marks it safe.
    By Lemma~\ref{lemma: round-agreement}, \(k\) is an ancestor of \(k'\), so \(\fn{Log}(k)\) is a prefix of \(\fn{Log}(k')\) and \(b\in\fn{Log}(k')\).
    Step \textbf{Deliver} then appends and tob-delivers every previously undelivered, non-\(\bot\) block of \(\fn{Log}(k')\), including \(b\) unless already tob-delivered.
\end{proof}

\subsection{Reliable Notification}

\input{security_analysis_btn}

\begin{lemma}
    \label{lemma: 2Delta-btn}
    If a correct party rn-confirms the round-\(r\) timeout flag at time \(t\), then all correct parties rn-confirm it by time \(\max(\text{GST}, t)+2\delta\).
\end{lemma}
\begin{proof}
    Suppose a party \party{i} rn-confirms the round-\(r\) timeout flag at time \(t\). This happens after receiving \(2f+1\) \(\Accept\) messages. At least \(2f+1-f= f+1\) of those were sent by correct parties. As such, they will be received by all other correct parties at the latest at time \(\max(\text{GST}, t)+\delta\). At this moment, these parties execute Step \textbf{Cascade} of the reliable notification protocol and send an \(\Accept\) message. Therefore, at least \(n-f\) parties will send this type of message by time \(\max(\text{GST}, t)+\delta\). One message delay later, at time \(\max(\text{GST}, t)+2\delta\), all correct parties receive these messages and rn-confirm the round-\(r\) timeout flag.
\end{proof}

\subsection{Latency Bounds}

\begin{lemma}[Good-case commit latency]
    \label{lemma: good-case-commit-latency}
    If round \(r\) in the non-speculative protocol is first entered by a correct party at or after GST and has a correct leader, all correct parties tob-deliver its proposed non-\(\bot\) block within \((\max(d_s,d_t)+1)\delta\) of its proposal. Thus the good-case commit latency of \(\shortname\) is at most \(4\delta\), and that of \(\shortname_{opt}\) is at most \(5\delta\).
\end{lemma}
\begin{proof}
    Let \(t_r\geq\text{GST}\) be the leader's proposal time. By~\Cref{lemma: optimistic2,lemma: deliver-information}, all correct parties enter round \(r\) and establish the proposal's safe parent and required disabled rounds by \(t_r+d_t\delta\). They rb-deliver the new proposal by \(t_r+d_s\delta\). Thus they mark round \(r\) safe by \(u=t_r+\max(d_s,d_t)\delta\).
    The proof of~\Cref{lemma: liveness-gst} shows that every correct party votes before its round-\(r\) timer expires, using the configured \(\Dto=(d_s+d_t)\Delta\) and \(\delta<\Delta\). In particular, each party has voted by \(u\): it has entered round \(r\) and marked it safe by then, and any party that has already advanced must already have voted.
    By \(u+\delta\), all correct parties receive the \(n-f\) commit votes needed for Step \textbf{Commit}. Since they already have \(\var{safe}[r]=\True\), Step \textbf{Deliver} also tob-delivers the proposed block by that deadline, unless it was already tob-delivered. Substituting the RBC parameters gives \(4\delta\) and \(5\delta\).
\end{proof}

\begin{lemma}[Optimistic commit latency]
    \label{lemma: optimistic-latency}
    Consider either the non-speculative or the speculative protocol, starting from its initial configuration with \(GST=0\). If at most \(f_o\) parties are Byzantine and the leaders of rounds \(1,\ldots,r\) are correct, all correct parties tob-deliver the non-\(\bot\) block proposed in round \(r\) within \((d_o+1)\delta\) of its proposal.
\end{lemma}
\begin{proof}
    Write \(t_j\) for the time at which \(L_j\) rb-broadcasts its proposal. We induct on \(j=1,\ldots,r\) to show that no correct party times out or rn-raises a timeout flag in round \(j\), its proposal names round \(j-1\) as parent, and every correct party marks round \(j\) safe and votes to commit it by \(t_j+d_o\delta\).

    Suppose the claim holds for all earlier rounds. RN Validity implies that no earlier round is disabled; nor can an earlier timeout abort round \(j\). An abort flag can only be raised following a genuine timeout in a lower round, so later rounds cannot invalidate this conclusion. Before the first round-\(j\) timer expires, no correct party can rn-raise its timeout flag, and RN Validity prevents Byzantine parties alone from disabling it.

    Let \(s_j\) be the first correct entry into round \(j\). For \(j=1\), all correct parties start in that round. For \(j>1\), the first entrant has marked every preceding round safe, since none is disabled. By RBC totality, all correct parties rb-deliver those preceding proposals by \(s_j+d_t\delta\). They can then mark the entire prefix safe, vote in any preceding rounds not yet traversed, and enter round \(j\). Consequently, \(L_j\) has proposed by \(s_j+d_t\delta\): either it already proposed speculatively or it proposes upon entry.

    A normal proposal names \(j-1\) as parent because no earlier round is disabled; a speculative proposal explicitly names \(j-1\). For \(j>1\), either proposal requires that \(L_{j-1}\) has already rb-broadcast its proposal, so \(t_1\leq\cdots\leq t_j\). Under the fault bound \(f_o\), all correct parties rb-deliver every proposal through round \(j\) by \(u_j=t_j+d_o\delta\). This deadline satisfies
    \[
        u_j\leq s_j+(d_t+d_o)\delta
        \leq s_j+(d_t+d_s)\delta
        <s_j+\Dto.
    \]
    Thus no round-\(j\) timer has expired by \(u_j\). Starting from genesis, all correct parties establish the safe chain through \(j\) and, by Steps \textbf{Vote} and \textbf{Advance round}, enter round \(j\) and vote by \(u_j\). This also prevents any subsequent genuine timeout in round \(j\), completing the induction. The argument includes startup rounds and proposals made before their leader enters the round; it requires no fixed spacing between proposals.

    Applying the induction to \(r\), all correct parties have \(\var{safe}[r]=\True\) and have sent a commit vote by \(t_r+d_o\delta\). The system's fault bound guarantees at least \(n-f\) correct parties, so every correct party receives a commit quorum by \(t_r+(d_o+1)\delta\). Steps \textbf{Commit} and \textbf{Deliver} then commit round \(r\) and tob-deliver its proposed block by that deadline, unless the block was already tob-delivered.
\end{proof}

\begin{corollary}[Optimistic commit latency with Bracha-RBC]
    \label{corollary: optimistic-latency-bracha}
    The optimistic commit latency (\Cref{def: optimistic-commit-latency}) of protocols \(\shortname\) and \(\shortname^s\) is at most \(4\delta\), with \(f_o=f\).
\end{corollary}

\begin{corollary}[Optimistic commit latency with Opt-RBC]
    \label{corollary: optimistic-latency-opt}
    The optimistic commit latency (\Cref{def: optimistic-commit-latency}) of protocols \(\shortname_{opt}\) and \(\shortname_{opt}^s\) is at most \(3\delta\), with \(f_o=\floor{n/2}-f\).
\end{corollary}

\begin{lemma}[Optimistic amortized block time]
    \label{lemma: optimistic-block-time-delivery}
    The optimistic amortized block time (\Cref{def: optimistic-amortized-block-time}) of \name using propose-upon-delivery pipelining is at most \(d_o\delta\).
\end{lemma}
\begin{proof}
    Write \(t_j\) for the round-\(j\) proposal time. Under the optimistic conditions, every correct party marks round \(j\) safe and votes by \(t_j+d_o\delta\) (\Cref{lemma: optimistic-latency}). Thus the next leader proposes by that time, giving \(t_{j+1}-t_j\leq d_o\delta\). Bounded commit latency does not affect the long-run rate, proving the claim.
\end{proof}

\begin{corollary}
    \label{corollary: optimistic-block-time-delivery}
    The optimistic amortized block time of \name variants \(\shortname\) and \(\shortname_{opt}\) is at most \(3\delta\) and \(2\delta\), respectively.
\end{corollary}

\begin{lemma}
    \label{lemma: bad-round-duration}
    If all correct parties have entered round \(r\) by time \(t\geq\text{GST}\), then all correct parties enter round \(r+1\) by time \(t+\Dto+d_t\delta\), where \(\Dto=(d_s+d_t)\Delta\).
\end{lemma}
\begin{proof}
    Set \(v=t+\Dto\). Each correct party starts its round-\(r\) timer on entry, so by \(v\) it has advanced to round \(r+1\), voted to commit round \(r\), or rn-raised its round-\(r\) timeout flag.

    If a correct party has marked round \(r\) safe or disabled by \(v\), RBC and RN totality propagate the corresponding evidence, including all safe ancestry and skipped-round flags, to every correct party by \(v+d_t\delta\). This follows from \(v\geq\text{GST}\) and \(d_t\geq 2\), even if some evidence predates GST. A party still in round \(r\) can then advance: a disabled round permits immediate advancement, and a safe round permits advancement once the party has voted or raised its timeout flag.

    Otherwise, no correct party has advanced from round \(r\) or voted in it by \(v\), so all correct parties have rn-raised its timeout flag. Their RN \(\Vote\) messages arrive by \(v+\delta\), and the resulting \(\Accept\) messages arrive by \(v+2\delta\). Thus RN unanimity disables the round everywhere by \(v+2\delta\leq v+d_t\delta\). In either case, all correct parties enter round \(r+1\) by \(t+\Dto+d_t\delta\).
\end{proof}

\begin{corollary}[Eventual worst-case view duration]
    \label{corollary: view-duration}
    The eventual worst-case view duration (\Cref{def: view-duration}) of \name is at most \(\Dto+d_t\delta\): \(5\Delta+2\delta\) for variant \(\shortname\) (Bracha-RBC) and \(8\Delta+4\delta\) for variant \(\shortname_{opt}\) (Opt-RBC).
\end{corollary}
\begin{proof}
    The views of \name are its rounds, and correct parties enter rounds in increasing order, one at a time (Step \textbf{Advance round}); hence the earliest time \(t_1\) at which every correct party has entered round \(r\) or higher is the time at which the last correct party enters round \(r\). For a round with \(t_1\geq\text{GST}\), Lemma~\ref{lemma: bad-round-duration} applied with \(t=t_1\) shows that every correct party enters a round strictly greater than \(r\) by time \(t_1+\Dto+d_t\delta\). The claim follows by instantiating \(d_s\) and \(d_t\) according to~\Cref{tab:broadcast_comparison}: \(\Dto=5\Delta\) and \(d_t=2\) for Bracha-RBC, and \(\Dto=8\Delta\) and \(d_t=4\) for Opt-RBC.
\end{proof}

\subsection{Speculative Pipelining}

We now turn to the version of \name appearing in~\Cref{fig:lionfish-protocol-speculative-pipelining}, which pipelines consecutive proposals.

The safety analysis of the standard version carries over because a correct party still never both votes to commit a round and rn-raises its timeout flag. Both \textbf{Timeout} and \textbf{Disable} check that the party has not voted in the target round, and set \(\var{timed\_out}\) or \(\var{aborted}\), respectively, which prevents later voting. In particular, \textbf{Disable} can abort the current round only if the party has not voted in it.

\begin{lemma}
    \label{lemma: liveness-gst-pipeline}
    With a round timer duration larger than \((d_t+d_s)\delta\), consider a window \(s,\ldots,s+2k-2\) of consecutive rounds with correct leaders such that round \(s\) is first entered by a correct party at or after GST. Every correct party commits the last round and obtains its safe proposal. Within this window, correct parties can execute Step \textbf{Timeout} only in the first \(k-1\) rounds.
\end{lemma}
\begin{proof}
    Inductively, parties keep advancing: if a round becomes safe or disabled at a correct party, totality propagates that information to everyone; otherwise, all correct parties eventually time out and RN unanimity disables the round. Since the window is first entered at or after GST, the totality argument of~\Cref{lemma: optimistic2} bounds the entry skew in each of its rounds by \(d_t\delta\).

    The key observation is that, after \(k\) consecutive rounds with correct leaders, the last proposal becomes safe before its round timer expires. Let \(j\) be that last round and \(t_j\) its proposal time. Follow its speculative parent links backwards. Either we reach a normal proposal, whose leader already knew its safe parent and any required disabled rounds, or we reach round \(j-k\), which \(j\)'s speculation guard requires to be safe. All intervening proposals have correct leaders and were broadcast by \(t_j\), since each speculative proposal follows receipt of its predecessor's initial message. Their RBCs finish everywhere by \(\max(\text{GST},t_j)+d_s\delta\), while the safe ancestry and skipped-round flags propagate by \(\max(\text{GST},t_j)+d_t\delta\). Since \(d_s\geq d_t\), round \(j\) is safe everywhere by \(\max(\text{GST},t_j)+d_s\delta\).

    Let \(u_j\geq\text{GST}\) be the first correct entry into \(j\). Its leader proposes by \(u_j+d_t\delta\), or has already proposed speculatively, so \(t_j\leq u_j+d_t\delta\). Writing \(T\) for the round timer duration, we have
    \[
        \max(\text{GST},t_j)+d_s\delta
        \leq u_j+(d_t+d_s)\delta < u_j+T.
    \]
    Thus \(j\) becomes safe everywhere before any correct timer expires, even if its proposal predates GST. Even an aborted party can advance once the proposal is safe, so no correct party times out in \(j\).

    Apply this observation to the last \(k\) rounds of our window, \(s+k-1,\ldots,s+2k-2\). Only its first \(k-1\) rounds can generate genuine timeouts. Earlier timeouts have index at most \(s+k-2\), so their cascades reach at most \(s+2k-3\). The final round is therefore neither aborted nor timed out at any correct party. Its safe proposal causes everyone to vote, and all correct parties commit within one further message delay.
\end{proof}

\begin{theorem}[Liveness]
    \label{theorem: liveness-pipelined}
    With a round timer duration larger than \((d_t+d_s)\delta\), and a leader schedule in which windows of \(2k-1\) consecutive rounds with correct leaders occur infinitely often, the protocol guarantees that if every correct party tob-submits infinitely many blocks, then every correct party tob-delivers infinitely many blocks tob-submitted by correct parties.
\end{theorem}
\begin{proof}
    Suppose every correct party tob-submits infinitely many blocks. By Corollary~\ref{corollary: higher-rounds} and the leader-schedule assumption, infinitely many windows of \(2k-1\) consecutive rounds with correct leaders have their first round first entered by a correct party at or after GST. By Lemma~\ref{lemma: liveness-gst-pipeline}, every correct party commits the last round of each such window, obtains its safe proposal, and tob-delivers its block and those of its ancestors through Step \textbf{Deliver}.

    Some correct party leads the last round of infinitely many of these windows. Since this party tob-submits infinitely many blocks, the block-selection rules ensure that these committed chains contain infinitely many of its submissions, as in the proof of~\Cref{theorem: liveness}. Thus every correct party tob-delivers infinitely many blocks tob-submitted by correct parties.
\end{proof}

\begin{theorem}[Totality]
    \label{theorem: totality-pipelined}
    Under the round-timer and leader-schedule assumptions of~\Cref{theorem: liveness-pipelined}, if a correct party tob-delivers a block \(b\), then every correct party eventually tob-delivers \(b\).
\end{theorem}
\begin{proof}
    A correct party tob-delivers \(b\) only upon committing a round \(r\) with a safe proposal and \(b\in\fn{Log}(r)\). Round advancement and the leader schedule provide a window of \(2k-1\) consecutive rounds with correct leaders above \(r\) whose first round is first entered by a correct party at or after GST. By~\Cref{lemma: liveness-gst-pipeline}, every correct party commits its last round \(r'>r\) and obtains its safe proposal.
    By~\Cref{lemma: round-agreement}, \(r\) is an ancestor of \(r'\), so \(b\in\fn{Log}(r')\). Step \textbf{Deliver} at every correct party therefore tob-delivers \(b\), unless that party has already tob-delivered it.
\end{proof}

\begin{lemma}[Optimistic amortized block time]
    \label{lemma: optimistic-block-time-pipelined}
    The optimistic amortized block time (\Cref{def: optimistic-amortized-block-time}) of \name variants \(\shortname^s\) and \(\shortname_{opt}^s\) is at most \(\delta\).
\end{lemma}
\begin{proof}
    Write \(t_j\) for the round-\(j\) proposal time. Under the optimistic conditions, the proof of~\Cref{lemma: optimistic-latency} shows that no round times out and every correct party marks round \(j\) safe by \(t_j+k\delta\), since \(d_o=k\) for both variants.

    For \(j>k\), leader \(L_j\) receives its predecessor's initial message by \(t_{j-1}+\delta\), and knows its safe anchor by \(t_{j-k}+k\delta\). It remembers an early message and proposes when both conditions hold, unless it has already proposed on entry. Thus
    \[
        t_j\leq\max\bigl(t_{j-1}+\delta,\;t_{j-k}+k\delta\bigr).
    \]
    Accounting for the first \(k\) startup rounds, induction gives
    \[
        t_j-t_1\leq (j-1)\delta+(k-1)^2\delta.
    \]
    The fixed startup cost and bounded commit latency do not affect the long-run rate, proving the claim.
\end{proof}

\begin{lemma}
    \label{lemma: bad-round-duration-pipelined}
    In the pipelined-proposal version of \name, if all correct parties have entered round \(r\) by time \(t\geq\text{GST}\), then all correct parties enter round \(r+1\) by time \(t+\Dto+d_t\delta\).
\end{lemma}
\begin{proof}
    The proof of~\Cref{lemma: bad-round-duration} carries over. A speculative abort already rn-raises the round's timeout flag and, like a timeout, permits advancement once the round is safe. Hence, by \(v=t+\Dto\), every correct party has advanced, voted, or rn-raised its round-\(r\) timeout flag, and both cases of that proof apply. The argument does not require abort flags alone to be confirmed: it uses totality if any correct party knows the round safe or disabled, and RN unanimity otherwise.
\end{proof}

\begin{corollary}[Eventual worst-case view duration, pipelined]
    \label{corollary: view-duration-pipelined}
    The eventual worst-case view duration (\Cref{def: view-duration}) of the pipelined-proposal version of \name is at most \(\Dto+d_t\delta\): \(5\Delta+2\delta\) for variant \(\shortname^s\) (Bracha-RBC) and \(8\Delta+4\delta\) for variant \(\shortname_{opt}^s\) (Opt-RBC).
\end{corollary}
\begin{proof}
    Identical to the proof of Corollary~\ref{corollary: view-duration}, using Lemma~\ref{lemma: bad-round-duration-pipelined} in place of Lemma~\ref{lemma: bad-round-duration}.
\end{proof}

%% file: security_analysis_btn.tex
\begin{lemma}
    Figure~\ref{fig:tbn} implements a reliable notification protocol.
\end{lemma}
\begin{proof}
\textbf{Unanimity. } If all correct parties rn-raise \(e\), all of them will send a $\sig{\Vote, e}$ message. Since there are at most $f$ Byzantine parties, all correct parties will eventually receive such messages from $n-f$ parties, causing them to send an $\sig{\Accept, e}$ message. Similarly, all correct parties will receive such messages from at least $n-f\geq 2f+1$ parties, causing them to rn-confirm \(e\).

\textbf{Totality. } Suppose that a correct party rn-confirms \(e\). For this to happen, it must have received $2f+1$ $\sig{\Accept, e}$ messages. Among all the parties that sent them, there are at least $f+1$ that are correct, meaning that the $\sig{\Accept, e}$ messages they sent will eventually be received by all the other participants. This will trigger all correct parties to send an $\sig{\Accept, e}$ message, causing all of them to receive such messages from $n-f$ parties, triggering the rn-confirmation of~\(e\).

\textbf{Validity.} Suppose a correct party rn-confirms \(e\). For this to happen, it must have received $\sig{\Accept, e}$ messages from $2f+1$ parties, meaning that at least $f+1$ correct parties sent this type of message. Consider the first correct party that sent such a message. Since $\sig{\Accept, e}$ messages are sent after receiving $\sig{\Accept, e}$ messages from $f+1$ parties or $\sig{\Vote, e}$ messages from $n-f$ parties, and there are at most $f$ Byzantine parties, the first correct party could not have received $f+1$ $\sig{\Accept, e}$ messages. Thus, it must have sent the message after receiving $\sig{\Vote, e}$ messages from $n-f$ parties. Among those, there must be at least $n-2f$ correct parties, which must have rn-raised \(e\).

\end{proof}

%% file: evaluation.tex
\section{Empirical Evaluation}%
\label{sec: evaluation}

We evaluate two prototypes, Mempool-\name and Dispersed-\name, on a geo-distributed testbed.
Our prototypes evaluate dissemination and ordering of synthetic transactions; application execution and durable recovery are outside the evaluation.
Mempool-\name uses Autobahn's shared mempool in signature-free form; Dispersed-\name uses leader-driven erasure-coded dissemination.
We compare each version against baselines using the same data-dissemination approach: Autobahn~\cite{giridharanAutobahnSeamlessHigh2024} and Sailfish++~\cite{shresthaOptimisticSignatureFreeReliable2025a} for Mempool-\name, and DispersedSimplex~\cite{shoupSingSongSimplex2024} for Dispersed-\name.
Both variants use either Bracha-RBC~\cite{bracha1987asynchronous} or Opt-RBC~\cite{shrestha2025optimistic}.
We also compare against PQ-DispersedSimplex, which replaces DispersedSimplex's BLS signatures with ML-DSA-65.
For the \name variants, we measure bandwidth usage and throughput recovery after one skipped proposal.

Mempool-\name achieves higher throughput and lower measured latency than Autobahn.
Dispersed-\name matches DispersedSimplex's peak throughput, with at most 10 percent higher latency before saturation, whereas the obvious alternative of using post-quantum-secure signatures (PQ-DispersedSimplex) instead incurs a $2\times$ latency cost.
Moreover, we find that the bandwidth/CPU costs of signature-freedom are negligible relative to data dissemination for leader-based protocols.

\paragraph{Implementation details.} Our source code and experiment scripts are available in a single artifact~\cite{simple-it-artifact}.
\Cref{tab:implementation-builds} summarizes our eight builds; the artifact also includes the two upstream baselines, Autobahn and Opt-Sailfish++.
Artifact directory names follow these build names, using lowercase and replacing the superscript dagger with \texttt{-failure}.

Mempool-\name is built by modifying the Autobahn codebase~\cite{autobahn-artifact}.
The original Autobahn's mempool takes 3 message delays. Parties broadcast transaction batches. After gathering $f+1$ votes for the batch, they forward the batch digest along with the votes. Leaders then propose the certified batch that has been seen by at least $n-f$ parties, as consensus input.
Mempool-\name reuses Autobahn's dissemination layer in a signature-free form. The leader's proposal includes a set of batches, each of which has received $2f+1$ batch votes (these votes are broadcast and are not embedded in the proposal). Upon receiving the proposal, a non-leader checks whether it has seen at least $f+1$ batch votes for every batch in the proposal. If so, the non-leader casts its own consensus vote for the proposal.
Thus Mempool-\name's mempool takes 2 message delays.
The consensus protocol is replaced with \name, following \Cref{fig:lionfish-protocol-prose}. Since Autobahn does not pipeline, we compare it with non-pipelined Mempool-\name.

Dispersed-\name and DispersedSimplex are built by modifying the Sailfish codebase~\cite{sailfish-impl}.
Our DispersedSimplex baseline implements only the failure-free path and is used only in fault-free comparisons.
PQ-DispersedSimplex is built on top of DispersedSimplex by replacing BLS signatures with ML-DSA-65 signatures~\cite{fips204} (implemented using the Rust \texttt{ml-dsa} crate~\cite{ml-dsa-crate}).
Dispersed-\name follows the pipelined protocol in \Cref{fig:lionfish-protocol-speculative-pipelining}, and all other dispersed baselines are pipelined.
In the Dispersed-\name implementations, leaders disseminate blocks using the erasure-coded RBC techniques of Cachin and Tessaro~\cite{cachin2005asynchronous} and Shrestha et al.~\cite{shrestha2025optimistic}.

A superscript dagger marks the separate builds used in the failure experiments.

\begin{table}[htbp]
\centering
\caption{Implementation builds included in our artifact.}
\label{tab:implementation-builds}
\footnotesize
\begin{tabularx}{\columnwidth}{@{}l>{\raggedright\arraybackslash}X@{}}
\toprule
\textbf{Implementation} & \textbf{Description} \\
\midrule
Bracha-Mempool-\name & Shared mempool; Bracha-RBC; no speculative proposals. \\[0.5ex]
Opt-Mempool-\name & Shared mempool; Opt-RBC; no speculative proposals. \\[0.5ex]
Opt-Mempool-\name{}\textsuperscript{\dag} & Separate build used for the mempool failure experiment. \\
\midrule
Bracha-Dispersed-\name & Erasure-coded dissemination; Bracha-RBC; speculative proposals with $k=3$. \\[0.5ex]
Opt-Dispersed-\name & Erasure-coded dissemination; Opt-RBC; speculative proposals with $k=2$. \\[0.5ex]
Opt-Dispersed-\name{}\textsuperscript{\dag} & Separate build used for the dispersed failure experiment. \\
\midrule
DispersedSimplex & Pipelined Simplex with erasure-coded dissemination and BLS signatures; failure-free path only. \\[0.5ex]
PQ-DispersedSimplex & DispersedSimplex with ML-DSA-65 replacing BLS signatures. \\
\bottomrule
\end{tabularx}
\end{table}

\paragraph{Experimental setup.} We conducted our evaluations on the Google Cloud Platform (GCP), deploying nodes evenly across five distinct regions: us-east1-b (South Carolina), us-west1-a (Oregon), europe-west1-b (Belgium), europe-north1-b (Finland) and asia-northeast1-a (Japan). Each node runs a single party. We employed c2-standard-16 instances~\cite{gcp}, each featuring $16$ vCPUs, $64$ GB of memory, and up to $32$ Gbps network bandwidth. All nodes ran on Ubuntu $22.04$, and round-trip latencies between GCP regions range from roughly 30 ms to 260 ms (details appear in~\Cref{tab:ping-latencies}).

\begin{table}[ht]
\centering
\caption{Ping latencies (in ms) between GCP regions}
\footnotesize
\label{tab:ping-latencies}
\begin{threeparttable}
\begin{tabular}{l|c r r r r}
\toprule
\multicolumn{1}{c|}{} & \multicolumn{5}{c}{\textbf{Destination}*} \\
\midrule
\textbf{Source} & us-e1 & us-w1 & eu-w1 & eu-n1 & as-n1 \\
\midrule
us-east1-b              & 0.70 & 64.17 & 92.21 & 111.95 & 157.79 \\
us-west1-a              & 64.30 & 0.84 & 130.67 & 159.78 & 88.57 \\
europe-west1-b          & 92.22 & 130.69 & 0.59 & 31.17 & 220.67 \\
europe-north1-b         & 111.96 & 159.77 & 31.16 & 0.72 & 262.90 \\
asia-northeast1-a       & 157.07 & 90.94 & 220.60 & 262.86 & 0.93\\
\bottomrule
\end{tabular}
\begin{tablenotes}
\small
\item[*] Region names are abbreviated versions of the source regions.
\end{tablenotes}
\end{threeparttable}
\end{table}

We use synthetic $512$-byte transactions and dispersed blocks of up to $50$ MB.
Runs last $120$ seconds with $50$ nodes, or $60$ seconds with $10$ nodes.
The main comparisons use $n=50$, a common size in prior evaluations~\cite{danezis2022narwhal,gelashvili2022jolteon,spiegelman2022bullshark,spiegelman2023shoal,shrestha2024sailfish,shresthaOptimisticSignatureFreeReliable2025a,babel2023mysticeti}, facilitating comparison at the same system size.

\paragraph{Metrics.} The \name latency benchmarks run without clients, generating payloads at block creation: each node creates mempool blocks, whereas leaders create dispersed proposals.
Both measure \emph{consensus-path latency}, from block creation to recorded commitment.
For Mempool-\name, this includes block dissemination and voting, waiting for inclusion in a leader's proposal, and consensus; for Dispersed-\name, the clock starts at proposal creation, excluding the earlier transfer-to-leader and inclusion stages.
Neither includes client submission or buffering before block creation, so these measurements are not end-to-end transaction latencies.
For both, latency is the mean over committed blocks of the elapsed time to the 25th-earliest recorded commit timestamp among the $50$ nodes.
This statistic does not capture completion at slower replicas, whereas our theoretical bounds require tob-delivery by all correct parties.
Throughput is measured by the number of transactions in tob-delivered blocks per second.

\paragraph{Methodology.} We increase the offered load to obtain the latency-throughput curves in \Cref{fig:mempool-evaluation,fig:erasure-coding-evaluation}. Subsequent measurements use the last operating point before saturation, where further load mainly increases latency.

\paragraph{Performance comparison under fault-free cases.}
\begin{figure*}[htbp]
  \centering
  \begin{minipage}{0.32\textwidth}
    \centering
    \includegraphics[width=\linewidth]{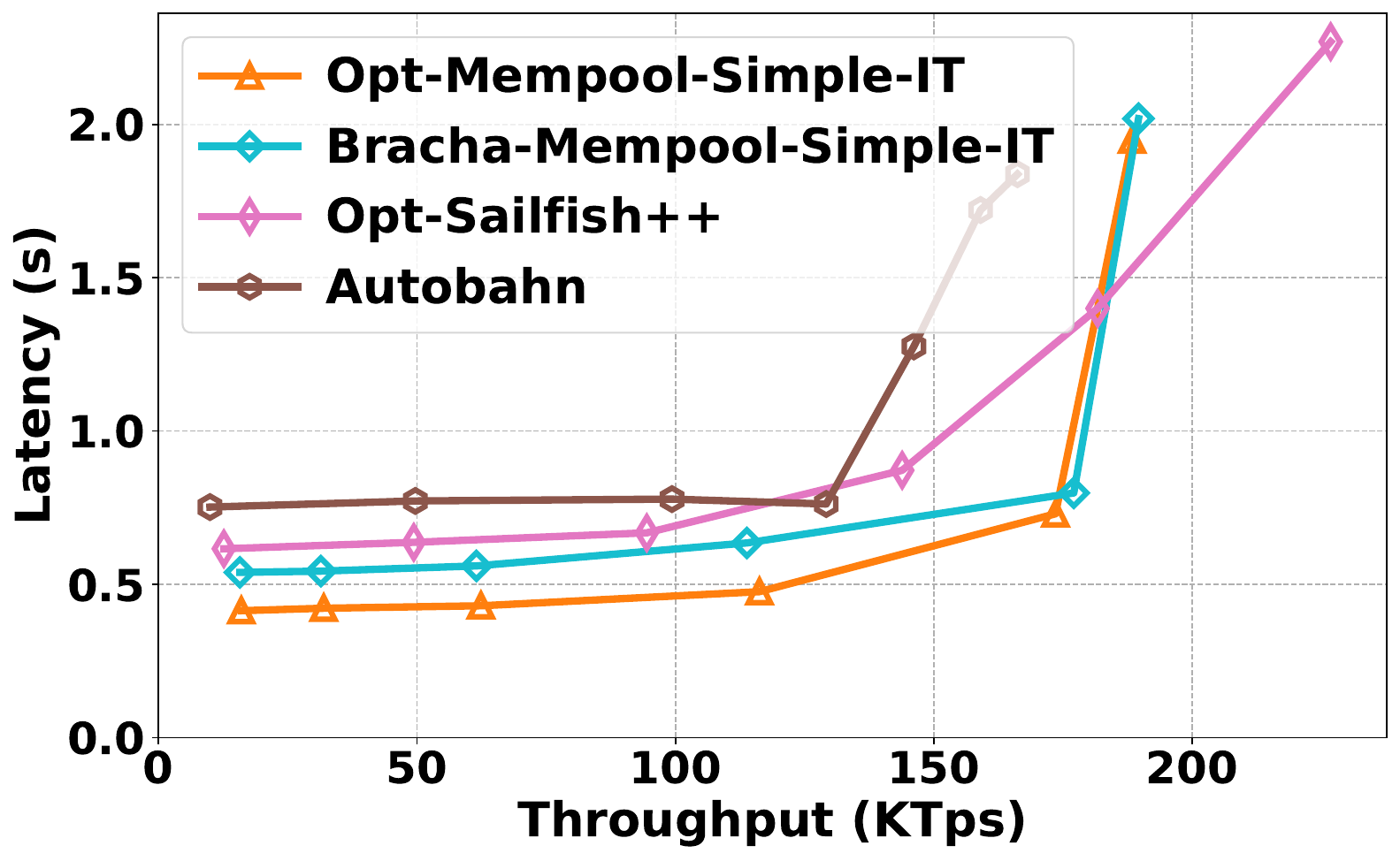}
    \caption{Latency vs. throughput of the Mempool-\name variants (consensus-path latency), Autobahn, and Opt-Sailfish++ with $50$ nodes.}
    \label{fig:mempool-evaluation}
  \end{minipage}
  \hfill
  \begin{minipage}{0.32\textwidth}
    \centering
    \includegraphics[width=\textwidth]{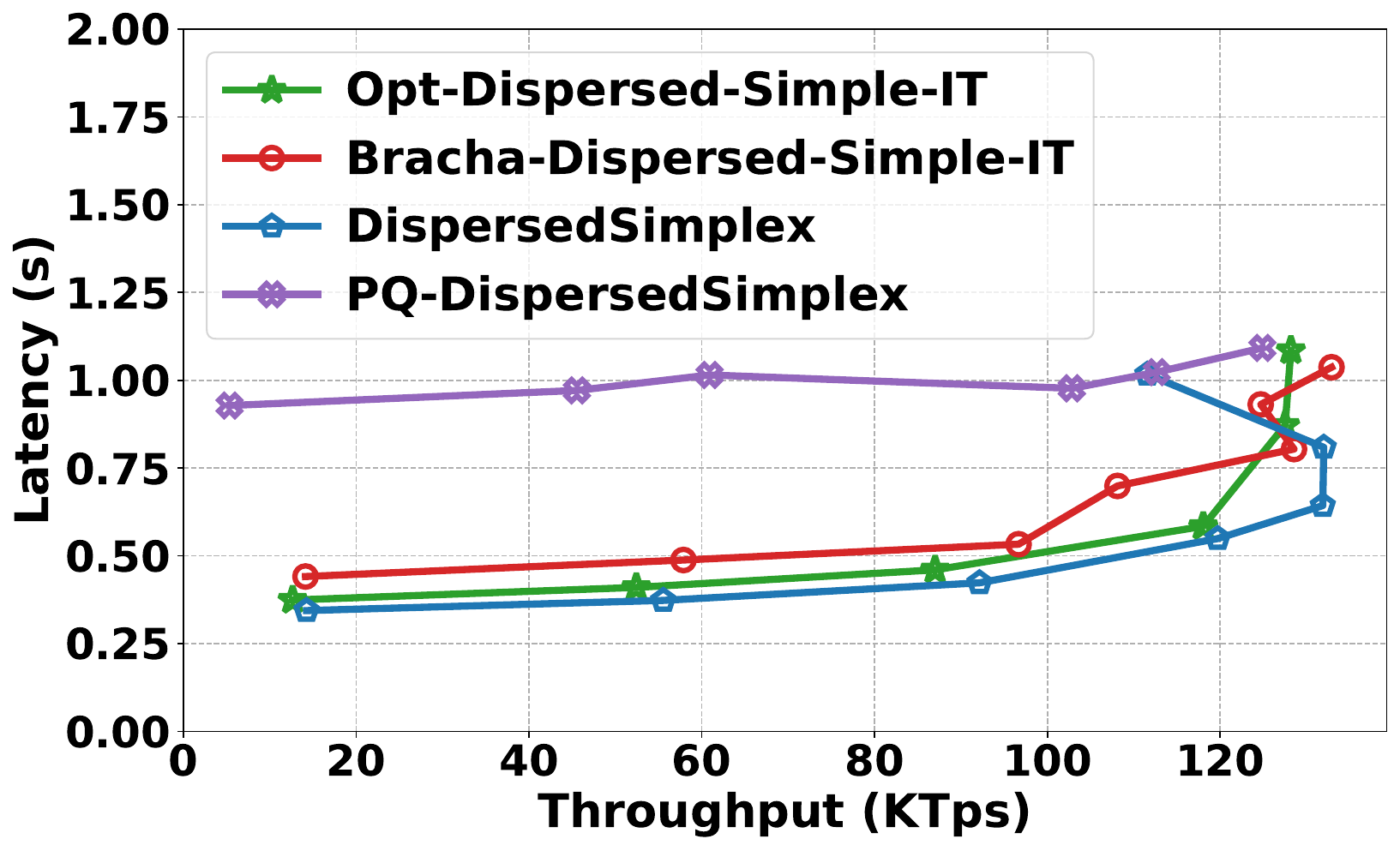}
    \caption{Consensus-path latency vs. throughput of the Dispersed-\name variants, DispersedSimplex, and PQ-DispersedSimplex with $50$ nodes.}
    \label{fig:erasure-coding-evaluation}
  \end{minipage}
  \hfill
  \begin{minipage}{0.32\textwidth}
    \centering
    \includegraphics[width=\textwidth]{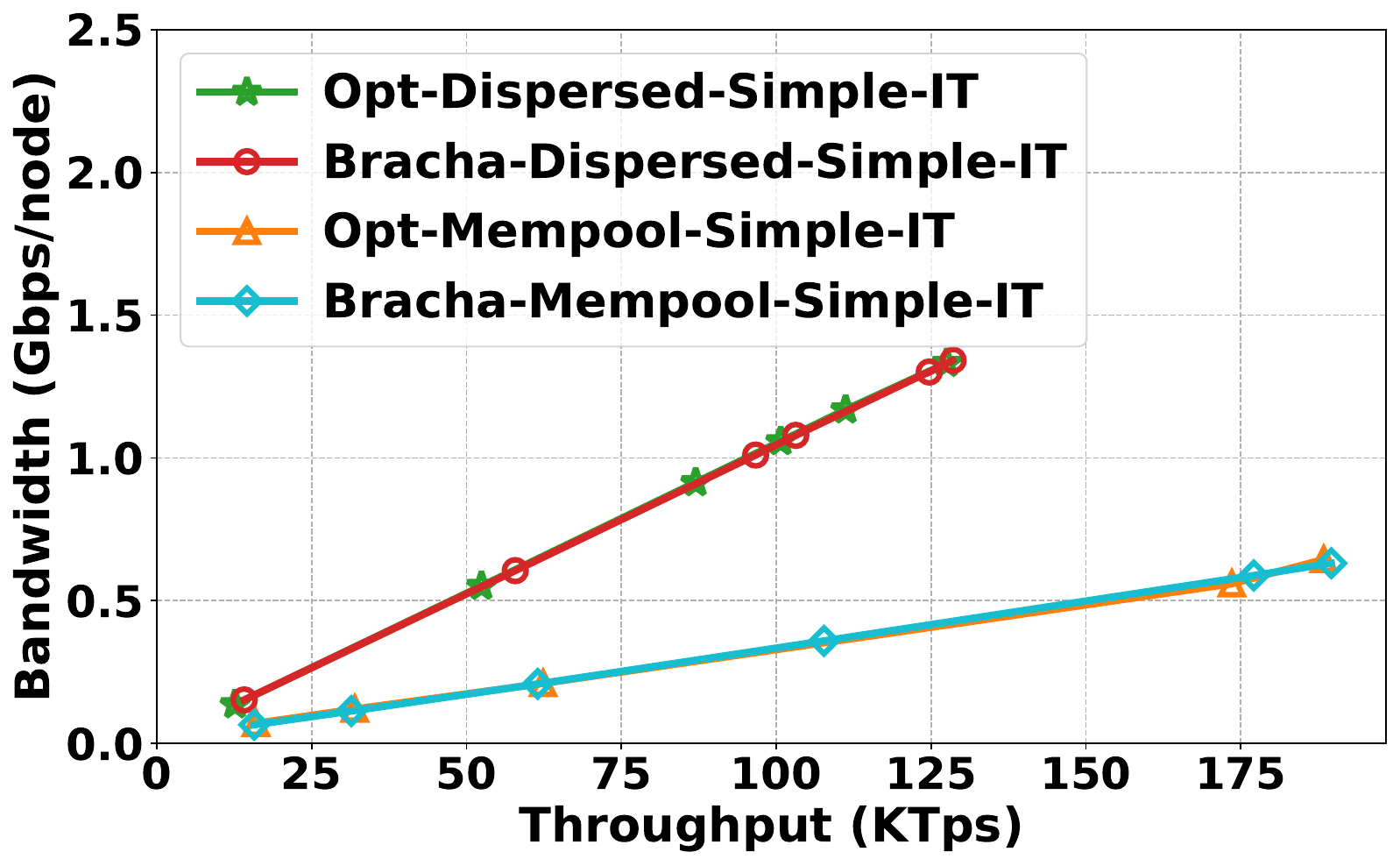}
    \caption{Outgoing bandwidth per node (Gbps/node) vs. throughput for the \name variants with $50$ nodes.}%
    \label{fig:bandwidth-evaluation}
  \end{minipage}
\end{figure*}
\Cref{fig:mempool-evaluation} compares Mempool-\name with Autobahn and with Opt-Sailfish++, a Sailfish++ baseline using Opt-RBC.
Autobahn uses one client per machine. For Autobahn and Mempool-\name, we measure latency from block creation, excluding preceding client-submission and batching delays.
In Opt-Sailfish++, a leader's block commits in three message delays, whereas a non-leader block needs at least five. This gives an estimated average of about $5.6$ message delays (see Footnote~\ref{fn:sailfish-latency}).
Autobahn requires 3 message delays for its mempool phase and 5 for consensus, totaling 8.
Mempool-\name's mempool phase takes 2 delays, giving totals of 6 for Bracha-Mempool-\name and, under the optimistic conditions of~\Cref{def: optimistic-commit-latency}, 5 for Opt-Mempool-\name, excluding waiting for inclusion in a leader's proposal.

\Cref{fig:erasure-coding-evaluation} compares the dispersed implementations, where in each round the leader disseminates its proposed block with erasure-coded RBC.
Because Dispersed-\name and DispersedSimplex use the same dissemination path, their latency gap primarily reflects the consensus path.
Bracha-Dispersed-\name has higher latency than Opt-Dispersed-\name because Bracha-RBC adds one message delay.
Before saturation, DispersedSimplex has only a modest latency advantage over Opt-Dispersed-\name, about $35$ ms on average.
In the $50$-node setting, this gap comes from two factors. First, Opt-Dispersed-\name waits for $40$ votes on its optimistic path, whereas DispersedSimplex needs a quorum of $34$ votes. Second, DispersedSimplex's signed quorum certificates can be forwarded directly. This could be beneficial under network congestion.
With $10$ nodes, this threshold gap disappears: both thresholds are $7$, so Opt-Dispersed-\name performs similarly to DispersedSimplex.
At about $1$ second latency, PQ-DispersedSimplex has roughly double the latency of even Bracha-Dispersed-\name, the slowest \name variant in this figure.

\Cref{fig:bandwidth-evaluation} reports outgoing bandwidth per node, measured in Gbps/node.
Bandwidth grows almost linearly with throughput, and variants using the same dissemination approach have nearly overlapping curves.
This indicates that data dissemination dominates the network cost.
At their respective high-throughput points, the mempool variants use about $0.6$ Gbps/node near $170$ KTps, whereas the dispersed variants use about $1.3$ Gbps/node near $120$ KTps.
Thus, in our implementation, the shared-mempool design has a substantially lower outgoing-bandwidth cost per transaction in a tob-delivered block than leader-driven erasure-coded RBC, since erasure-coded dispersal introduces a constant-factor communication overhead for each payload.

\paragraph{Recovery after a skipped proposal.}
\begin{figure}[t]
    \centering
    \includegraphics[width=0.32\textwidth]{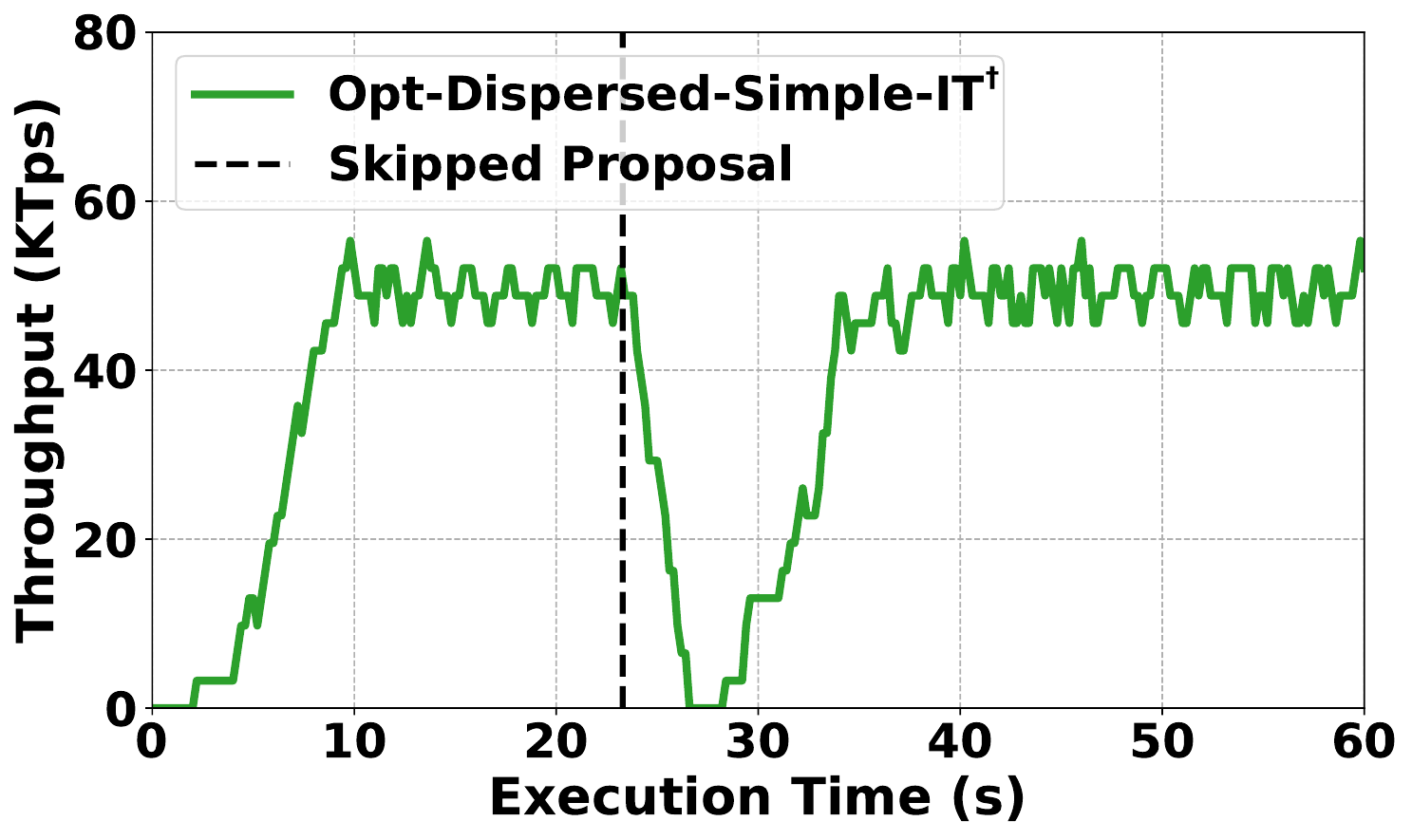}
    \caption{Throughput after one skipped proposal for Opt-Dispersed-\name{}\textsuperscript{\dag} with $10$ nodes (separate failure build).}
    \label{fig:failure-erasure-evaluation}
\end{figure}
\begin{figure}[t]
    \centering
    \includegraphics[width=0.32\textwidth]{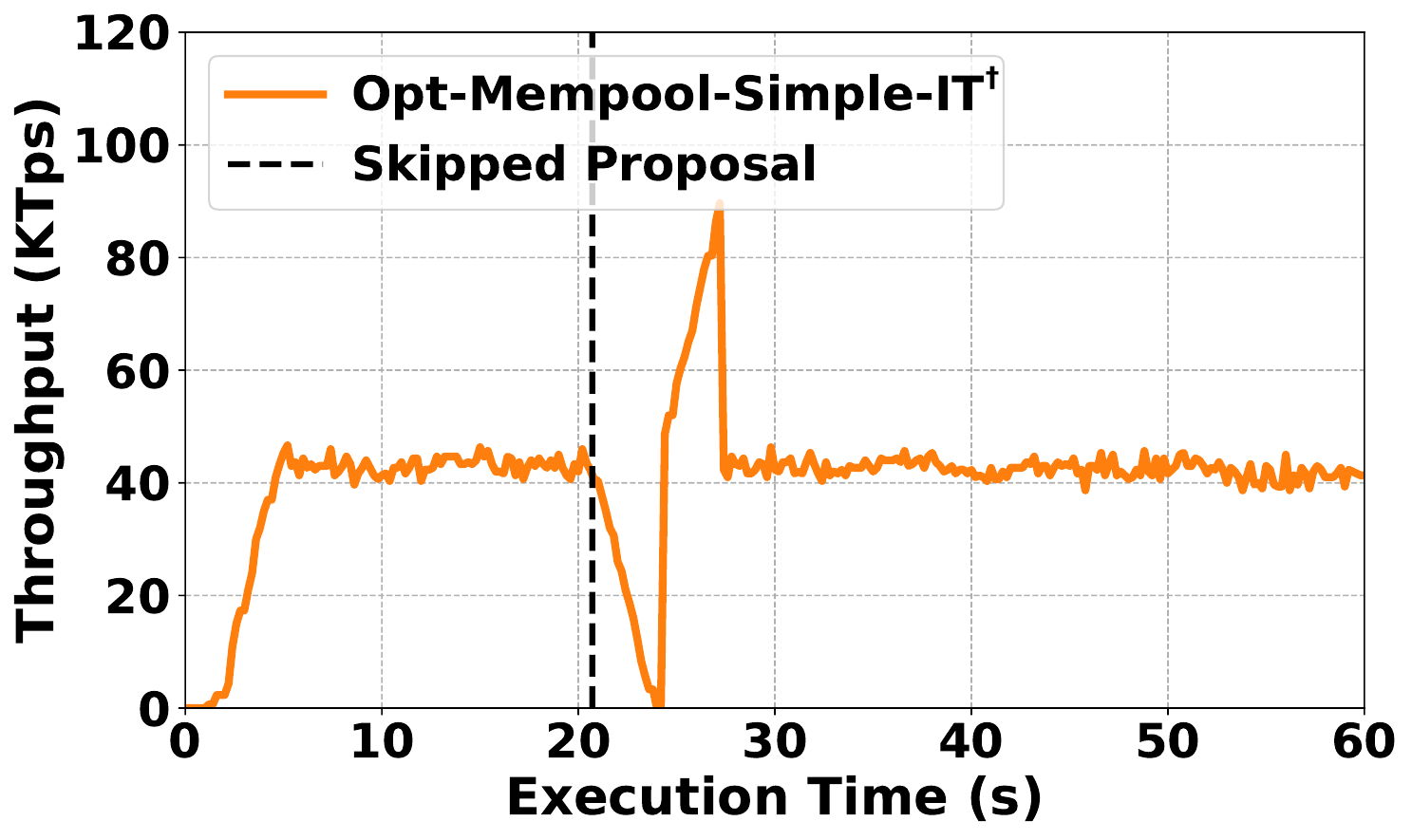}
    \caption{Throughput after one skipped proposal for Opt-Mempool-\name{}\textsuperscript{\dag} with $10$ nodes (separate failure build).}
    \label{fig:failure-mempool-evaluation}
\end{figure}
We use the separate Opt-Mempool-\name{}\textsuperscript{\dag} and Opt-Dispersed-\name{}\textsuperscript{\dag} builds with $10$ nodes to measure throughput recovery from a single missed proposal.
At about $20$ seconds, we skip one proposal by the leader in europe-west1-b (Belgium); it participates normally in subsequent rounds.
The round timer is $3$ seconds, and throughput is computed over a $3$-second sliding window.

In \Cref{fig:failure-erasure-evaluation}, throughput drops at the skipped proposal and recovers after roughly $10$ seconds. In \Cref{fig:failure-mempool-evaluation}, the drop is followed by a spike: the mempool keeps producing batches while the proposal is skipped, so the next tob-delivery covers more batches and temporarily raises throughput.

\paragraph{Discussion on erasure-coding costs.}
To evaluate the cost of erasure-coded dissemination, \Cref{tab:ec-coding-cost} reports the measured CPU-side cost of constructing and reconstructing erasure-coded proposals in the $50$-node deployment.
All runs use Reed-Solomon codes with $18$ data shards and $32$ parity shards.
Proposal construction includes block serialization, Reed-Solomon encoding, and constructing a Merkle tree over the encoded shards.
For reconstruction, we measure the time after a node has collected enough erasure-coded fragments to recover the block payload and verify that the recovered data matches the Merkle root in the header.
The latency column reports the consensus-path latency of Opt-Dispersed-\name with $50$ nodes.
These costs grow roughly linearly with the shard size.
\begin{table}[ht]
\centering
\caption{Erasure-coding costs for Dispersed-\name in the $50$-node deployment.}
\label{tab:ec-coding-cost}
\footnotesize
\setlength{\tabcolsep}{2.5pt}
\begin{threeparttable}
\begin{tabular}{@{}lcccc@{}}
\toprule
\makecell{\textbf{Block}\\\textbf{size}} &
\makecell{\textbf{Shard}\\\textbf{size}} &
\makecell{\textbf{Proposal construction}\\\textbf{p50/p95}} &
\makecell{\textbf{Reconstruction}\\\textbf{p50/p95}} &
\textbf{Latency} \\
\midrule
$10$ MB & $0.56$ MB & $50/54$ ms & $10/12$ ms & $460$ ms \\
$30$ MB & $1.69$ MB & $131/140$ ms & $26/31$ ms & $872$ ms \\
$50$ MB & $2.82$ MB & $212/227$ ms & $42/49$ ms & $1085$ ms \\
\bottomrule
\end{tabular}
\end{threeparttable}
\end{table}

%% file: related_work.tex
\section{Related Work}%
\label{sec: related work}

\paragraph{Post-Quantum and Signature-Free BFT.}
A direct way to make BFT protocols quantum-secure is to replace classical digital signatures with post-quantum signatures.
The natural candidates are the NIST-standardized ML-DSA~\cite{fips204}, the likely default for performance-oriented systems, and the hash-based SLH-DSA~\cite{fips205}, a more conservative but less efficient alternative.
However, post-quantum signatures are much larger than classical signatures, and some are also slower to verify, increasing the latency on the consensus critical path (as shown in~\Cref{fig:erasure-coding-evaluation}). This motivates removing public-key signatures from the BFT protocol entirely.

Signature-free BFT protocols rely only on authenticated channels after setup and are closely related to unauthenticated and information-theoretic BFT.
The earliest is Algorithm BFT from Castro's thesis~\cite[Chapter~3]{castro_thesis}, a PBFT variant that replaces signatures with MAC vectors, achieving $3\delta$ good-case latency at the cost of $O(n^3)$ worst-case communication per view.
IT-HS~\cite{abrahamInformationTheoreticHotStuff2021} achieves optimal resilience and responsiveness with $O(n^2)$ communication per view, but has $6\delta$ good-case latency. TetraBFT~\cite{yuTetraBFTReducingLatency2024b} reduces this to $5\delta$, while Forget-IT~\cite{abraham2026forget} achieves the optimal $3\delta$ latency.
Shoup~\cite[Section~7]{shoupSingSongSimplex2023a} also sketches a signature-free variant of DispersedSimplex using Bracha-style echo/ready certification, with a good-case latency of $6\delta$.
However, to our knowledge, these protocols have not been implemented or evaluated in high-performance settings.

\paragraph{Leader-Based BFT Consensus.}
Leader-based partially synchronous protocols~\cite{castro1999practical,yin2019hotstuff,gelashvili2022jolteon,chanSimplexConsensusSimple2023,shresthaHydrangeaOptimisticTwoRound2025,giridharanAutobahnSeamlessHigh2024,tonkikh2025raptr} achieve low latency by leveraging signatures. In particular, PBFT~\cite{castro1999practical}, Simplex~\cite{chanSimplexConsensusSimple2023}, Hydrangea~\cite{shresthaHydrangeaOptimisticTwoRound2025} and Raptr~\cite{tonkikh2025raptr} are protocols with optimal $3\delta$ good-case latency.

To improve throughput, many systems with leader-based protocols decouple data dissemination from ordering so that consensus runs on small digests rather than raw transaction data.
Examples include Narwhal~\cite{danezis2022narwhal} and Autobahn~\cite{giridharanAutobahnSeamlessHigh2024}. Another approach, used by DispersedSimplex~\cite{shoupSingSongSimplex2024}, employs leader-driven erasure-coded dissemination to balance bandwidth while preserving low latency. 

\paragraph{DAG-Based BFT Consensus.}
DAG-based protocols disseminate data in parallel and derive an ordering from the resulting DAG. Narwhal and Tusk~\cite{danezis2022narwhal} separate data availability from ordering, while Bullshark~\cite{spiegelman2022bullshark}, Shoal~\cite{spiegelman2023shoal}, Shoal++~\cite{balaji2024shoal}, Sailfish~\cite{shrestha2024sailfish} and Mysticeti~\cite{babel2023mysticeti} further reduce latency. Sailfish++~\cite{shrestha2025optimistic} extends Sailfish to be signature-free and achieves $3\delta$ latency for leader vertices under optimistic conditions, but incurs cubic message complexity per view.

\paragraph{Reliable Broadcast.}
Bracha-RBC~\cite{bracha1987asynchronous} tolerates optimal $f<n/3$ Byzantine faults and has $3\delta$ good-case latency. 
For long messages, erasure coding yields communication-efficient RBC: Cachin and Tessaro~\cite{cachin2005asynchronous} obtain $O(nL+\kappa n^2\log n)$ communication for an $L$-bit input.
Signature-free RBC can complete in two steps under a stronger fault tolerance assumption of $n\geq 4f$~\cite{abraham2022good}.
Shrestha et al.~\cite{shrestha2025optimistic} give Opt-RBC, an optimistic signature-free RBC that completes in two steps when at most \(f_o=\floor{n/2}-f\) parties are Byzantine, while preserving optimal resilience otherwise.

%% file: conclusion.tex
\section{Conclusion}%
\label{sec: conclusion}
This paper presents \name, a signature-free, leader-based BFT consensus protocol with optimal resilience that commits in 4 message delays in the good case (one more than the optimum) and 3 on its optimistic path. 
By replacing signature-based quorum certificates with reliable broadcast and a new reliable notification primitive, \name retains the simplicity of authenticated protocols such as Simplex, making it amenable to practical optimizations such as pipelining and speculative proposals that reduce its optimistic amortized block time to at most a single message delay.
We give the first practical implementation of a signature-free, leader-based BFT protocol, pairing it with the two block-dissemination approaches that dominate production systems. Our geo-distributed evaluation shows that both variants closely match or outperform their state-of-the-art, post-quantum-vulnerable counterparts.
These results demonstrate that post-quantum security in BFT consensus need not come at the cost of performance, and that signature-free protocols are a practical path to quantum-safe, high-performance replicated systems.